\documentclass[twocolumn]{revtex4-1}
\usepackage{amsmath}
\usepackage{amsthm}
\usepackage{latexsym}
\usepackage{amsfonts}
\usepackage{amssymb}
\usepackage{color}
\usepackage{bbm,dsfont}
\usepackage{graphicx}
\usepackage{MnSymbol}
\usepackage{enumerate}
\usepackage{bbding}
\usepackage{changes}
\usepackage{tikz}
\usepackage[caption=false]{subfig}
\usepackage{mathrsfs}

\newcommand{\bmu}{{\boldsymbol\mu}}
\newcommand{\bnu}{{\boldsymbol\nu}}
\newcommand{\bgam}{{\boldsymbol\gamma}}

\definechangesauthor[color=red]{cc}
\definechangesauthor[color=blue]{th}
\definechangesauthor[color=green]{at}

\newtheorem{proposition}{Proposition}
\newtheorem{theorem}{Theorem}
\newtheorem{lemma}{Lemma}

\newtheorem*{question*}{Question}

\newtheorem{definition}{Definition}

\newcommand{\Cb}{\mathbb C} 
\newcommand{\Rb}{\mathbb R} 
\newcommand{\Zb}{\mathbb Z} 
\newcommand{\mo}[1]{\left| #1 \right|} 
\newcommand{\abs}{\mo} 

\newcommand{\hi}{\mathcal{H}} 
\newcommand{\hh}{\mathcal{H}} 
\newcommand{\lhs}{\mathcal{L}_s(\hi)} 
\newcommand{\spanno}[1]{{\rm span}\,\left\{ #1 \right\}}

\newcommand{\ip}[2]{\left\langle\,#1\,|\,#2\,\right\rangle} 
\newcommand{\kb}[2]{|#1\rangle\langle#2|} 
\newcommand{\no}[1]{\left\|#1\right\|} 
\newcommand{\tr}[1]{{\rm tr}\left[#1\right]} 

\newcommand{\id}{\mathbbm{1}} 

\renewcommand{\rho}{\varrho}
\newcommand{\lam}{\lambda}

\newcommand{\rank}{{\rm rank}\,} 

\newcommand{\Ao}{\mathsf{A}}
\newcommand{\Bo}{\mathsf{B}}
\newcommand{\Mo}{\mathsf{M}}
\newcommand{\Po}{\mathsf{P}}
\newcommand{\Qo}{\mathsf{Q}}
\newcommand{\Uo}{\mathsf{U}}

\newcommand{\phii}{\varphi}

\newcommand{\en}{\mathcal{E}} 
\newcommand{\enf}{\mathcal{F}} 
\newcommand{\enp}{{\hat{\mathcal{E}}}} 
\newcommand{\Pg}{P_{{\rm guess}}}
\newcommand{\Ppg}{\Pg^{{\rm post}}}
\newcommand{\Prg}{\Pg^{{\rm prior}}}

\newcommand{\V}{\ca{V}}
\newcommand{\A}{C_0}
\newcommand{\B}{\overline{C}_0}
\newcommand{\CC}{C}

\newcommand{\ca}[1]{\mathcal{#1}} 
\newcommand{\wit}{\xi} 
\newcommand{\Wit}{\Xi} 
\newcommand{\De}[1]{{\mathcal{D}}(#1)} 
\newcommand{\OO}{\ca{O}_{X,Y}} 
\newcommand{\cOO}{\ca{O}^{\mathrm{com}}_{X,Y}} 
\newcommand{\iOO}{\ca{O}^{\mathrm{inc}}_{X,Y}} 
\definecolor{darkgreen}{rgb}{0,0.6,0.2}
\definecolor{darkyellow}{rgb}{1,0.6,0}

\newcommand{\aligno}[1]{\begin{align*} #1 \end{align*}}
\newcommand{\aligsi}[1]{\begin{align} #1 \end{align}}
\newcommand{\equano}[1]{\begin{equation*} #1 \end{equation*}}
\newcommand{\equasi}[2]{\begin{equation} \label{#1} #2 \end{equation}}

\newcommand{\pair}[2]{\langle\,#2\,,\,#1\,\rangle} 
\newcommand{\Cscr}{C}
\newcommand{\Fscr}{\mathcal{F}}
\newcommand{\Mscr}{\mathcal{M}}
\newcommand{\Oscr}{\mathcal{O}}
\newcommand{\Sscr}{\mathcal{S}}

\newcommand{\fun}[1]{\Fscr( #1 )} 
\newcommand{\funo}[2]{\Fscr_{#1}( #2 )} 
\newcommand{\obs}[1]{\Oscr( #1 )} 

\newcommand{\Aff}[1]{\Mscr(#1)}
\newcommand{\ri}[1]{{\rm ri}(#1)}

\begin{document}

\title[]{Quantum Incompatibility Witnesses}

\author{Claudio Carmeli}
\email{claudio.carmeli@gmail.com}
\affiliation{DIME, Universit\`a di Genova, Via Magliotto 2, I-17100 Savona, Italy}

\author{Teiko Heinosaari}
\email{teiko.heinosaari@utu.fi}
\affiliation{QTF Centre of Excellence, Turku Centre for Quantum Physics, Department of Physics and Astronomy, University of Turku, FI-20014 Turku, Finland}

\author{Alessandro Toigo}
\email{alessandro.toigo@polimi.it}
\affiliation{Dipartimento di Matematica, Politecnico di Milano, Piazza Leonardo da Vinci 32, I-20133 Milano, Italy}
\affiliation{I.N.F.N., Sezione di Milano, Via Celoria 16, I-20133 Milano, Italy}

\begin{abstract}
We demonstrate that quantum incompatibility can always be detected by means of a state discrimination task with partial intermediate information. 
This is done by showing that only incompatible measurements allow for an efficient use of premeasurement information in order to improve the probability of guessing the correct state. Thus, the gap between the guessing probabilities with pre- and postmeasurement information is a witness of the incompatibility of a given collection of measurements. We prove that all linear incompatibility witnesses can be implemented as some state discrimination protocol according to this scheme.
{As an application, we characterize the joint measurability region of two noisy mutually unbiased bases.}
\end{abstract}

\maketitle


{\em Introduction.---} Quantum incompatibility is one of the key features that separate the quantum from the classical world \cite{HeMiZi16}. 
It gives rise to several among the most intriguing quantum phenomena, including measurement uncertainty relations \cite{BuLaWe14}, contextuality \cite{LiSpWi11} and nonlocality \cite{Fine82}. 
So far, however, the direct experimental verification of quantum incompatibility has been a demanding task, as the known detection methods, based on Bell experiments \cite{WoPeFe09,BaGaGhKa13,KaGhChBa16} and steering protocols \cite{UoMoGu14,QuVeBr14,UoBuGuPe15,ChBuLiCh16}, rely on entanglement.

In this paper, we show that quantum incompatibility can be detected by means of a state discrimination task with partial intermediate information. 
More precisely, we consider a scenario where Alice sends Bob a quantum system that she has prepared into a state chosen from one of $n$ disjoint state ensembles, but she reveals to him the chosen ensemble only at a later time. 
Bob can then decide to perform his measurement either before or after Alice's announcement and, importantly, the achievable success probabilities can be compared.
We show that Bob can benefit from prior compared to posterior measurement information and improve his probability of guessing the correct state only if his measurements are incompatible.

Looking at it from another perspective, the difference between Bob's guessing probabilities with pre- and postmeasurement information is a witness of the incompatibility of the collection of measurements he uses in the discrimination task. 
Since the complement set of incompatible collections of measurements is the closed and convex set of all the compatible collections of measurements, this observation sets the previous detection scheme for incompatibility within the broader framework of witnesses.

In general, a witness is any experimentally assessable linear function whose value is greater than or equal to zero whenever the measured object does not have the investigated property, but gives a negative value at least for some object with that property. 
The paradigmatic example of witnesses is that of entanglement witnesses, which have become one of the main methods to detect entanglement \cite{GuTo09,HoHoHoHo09}.
Other examples include the detection of non-Gaussianity of states \cite{HuGeTuPaKi14}, dimensionality of correlations \cite{BrPiAcGiMeSc08}, or for the unital channels the detection of not being a random unitary channel \cite{MeWo09}. 
The fact that witnesses can be applied to detect incompatibility has been recently noted in \cite{Jencova18,BlNe18}.

We prove that any incompatibility witness essentially arises as a state discrimination task with intermediate information of the type described above. By standard separation results for convex sets, this implies that all incompatible sets of measurements can be detected by performing some state discrimination where premeasurement information is strictly better than postmeasurement information. 
This yields a novel operational interpretation of quantum incompatibility, and provides a method to detect it in a physically feasible experiment.
In particular, this proves that entanglement is not needed to reveal incompatibility.


{\em General framework of witnesses.---} We briefly recall the general setting of witnesses as this clarifies our main results on incompatibility witnesses and makes the reasoning behind them easy to follow.

Let $\V$ be a real linear space and $\CC\subset\V$ a compact convex subset that mathematically describes the objects we are interested in. 
This set is further divided into two disjoint subsets $\A$ and $\B$, with $\A$ being closed and convex. 
We can think of $\A$ and $\B$ as properties -- either an element $x\in \CC$ is in $\A$ or in $\B$.
A witness of the property $\B$, or \emph{$\B$-witness}, is a map $\wit:\CC \to \Rb$ such that 
\begin{enumerate}
\item[(W1)] $\wit(x) \geq 0$ for all $x\in \A$ and $\wit(x) < 0$ at least for some $x\in \B$;
\item[(W2)] $\wit(t x + (1-t)y)=t\wit(x) + (1-t)\wit(y)$ for all $x,y\in\CC$ and $t\in[0,1]$.
\end{enumerate}
By condition (W2), each witness generates a hyperplane separating \(\V\) into two half-spaces. Condition (W1) then asserts that one of the two halves entirely contains \(\A\), but still does not contain all of $\CC$ (see Fig.~\ref{fig:w}).

We say that an element $x\in \B$ is \emph{detected} by $\wit$ if $\wit(x)<0$, and we denote by $\De{\wit}$ the subset of all elements of $\B$ that are detected by $\wit$.
Another $\B$-witness $\wit'$ is called {\em finer} than $\wit$ if $\De{\wit'}\supseteq\De{\wit}$, and in this case we write $\wit\preccurlyeq\wit'$. 
If $\De{\wit'}=\De{\wit}$, we say that $\wit$ and $\wit'$ are \emph{detection equivalent} and denote this by $\wit\approx\wit'$ (see Fig.~\ref{fig:w}).
As we typically aim to detect as many elements as possible, we favor witnesses that cannot be made any finer. 
A necessary condition for $\wit$ being optimal in that sense is that $\wit$ is \emph{tight}, meaning that $\wit(x)=0$ for some $x\in\A$.

\begin{figure}
\centering
\includegraphics[width=0.40\textwidth]{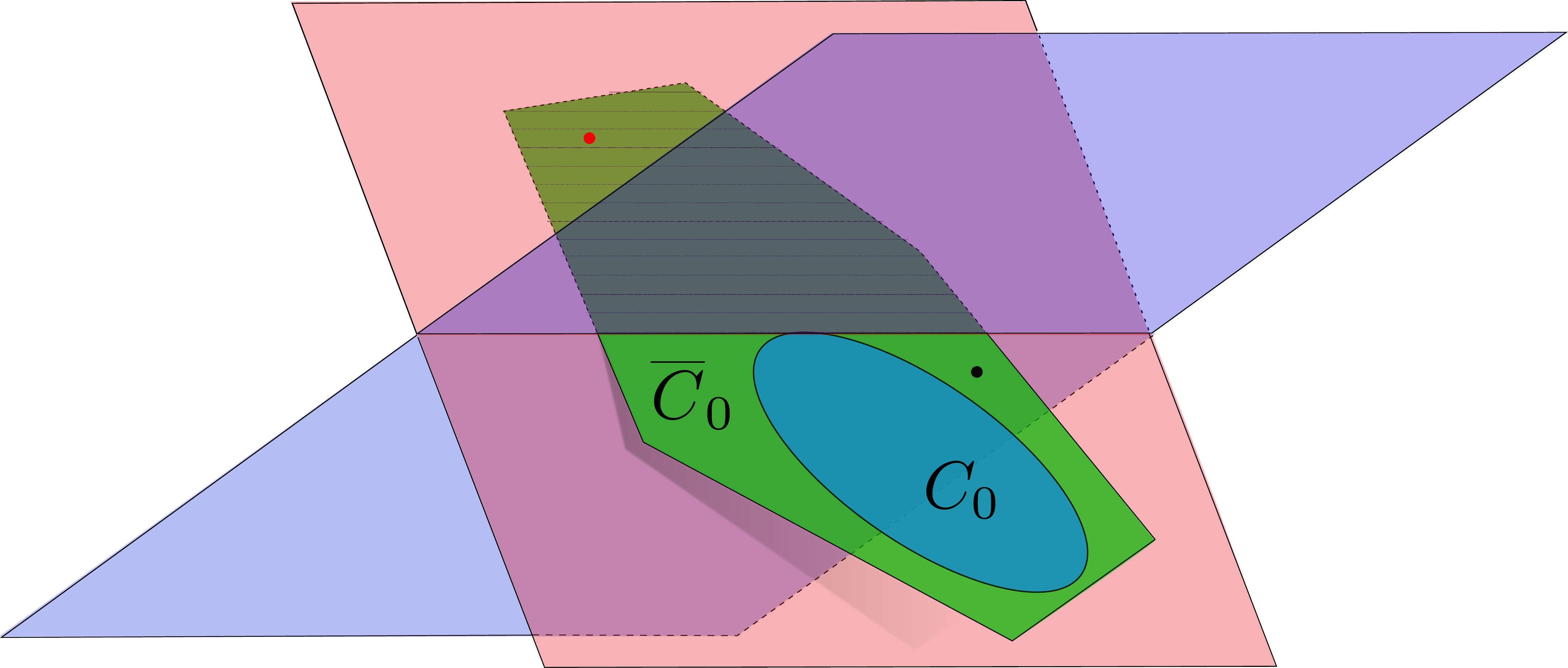}
\caption{Witnesses are associated with hyperplanes, and they are detection equivalent if they yield the same separation of the set $C$. Here, two tight equivalent witnesses detect the red point, but not the black one.\label{fig:w}}
\end{figure}

Any $\B$-witness $\wit$ can be written in the form
\begin{equation}\label{eq:witness_main}
\wit(x) = \delta - v^*(x) \qquad \forall x\in\CC \,,
\end{equation}
where $v^*:\V\to\Rb$ is a linear map and $\delta\in\Rb$ is a constant.
An essential point for our later developments is that the representation \eqref{eq:witness_main} of a witness $\wit$ is not unique but there is some freedom in the choice of $v^*$ and $\delta$.
 In addition, if we are only interested in the set of detected elements $\De{\wit}$, we have a further degree of freedom, coming from the possibility to switch from $\wit$ to an equivalent $\B$-witness $\wit' = \alpha\wit$ for some constant $\alpha>0$. 


{\em Detecting quantum incompatibility.---} A measurement with a finite outcome set $X$ is mathematically described as a positive operator valued measure (POVM), i.e., a map $\Ao$ from $X$ to the set $\lhs$ of self-adjoint linear operators on a Hilbert space $\hi$ such that the operators $\Ao(x)$ are positive (meaning that $\ip{\psi}{\Ao(x)\psi}\geq 0$ for all $\psi\in\hi$) and they satisfy the normalization condition $\sum_x \Ao(x) = \id$.

For clarity, we limit our discussion to pairs of measurements. 
The treatment of finite collections of measurements is similar. 
Two measurements $\Ao$ and $\Bo$, having outcome sets $X$ and $Y$, respectively, are \emph{compatible} if there exists a measurement $\Mo$, called their \emph{joint measurement}, with  outcome set $X\times Y$, such that
$\sum_y \Mo(x,y) = \Ao(x)$ and $\sum_x \Mo(x,y) = \Bo(y)$.
Otherwise, $\Ao$ and $\Bo$ are \emph{incompatible}.

By $\OO$ we denote the compact set of all pairs of measurements $(\Ao,\Bo)$ with outcome sets $X,Y$, respectively.
This set is divided into compatible pairs $\cOO$ and incompatible pairs $\iOO\equiv \overline{\cOO}$.
 We define convex combinations in $\OO$ componentwise, and it follows that  the subset $\cOO$ of compatible pairs is closed and convex. 
Hence we can consider $\iOO$-witnesses; we call them \emph{incompatibility witnesses} (IWs).


{\em Discrimination scenario as an incompatibility witness.---} In the standard state discrimination scenario \cite{Holevo73,YuKeLa75,QDET76}, Alice picks a label $z$ from a given set $Z$ with probability $p(z)$. 
She encodes the label into a quantum state $\varrho_z$ and delivers the state to Bob.
Bob knows the set $\{\varrho_z\}_{z\in Z}$ of states used in the encoding.
He is trying to recover the label by making a measurement on the quantum system that he has received.
It is convenient to merge the a priori probability distribution $p$ and the state encoding into a single map $\en$, given as $\en(z)=p(z)\varrho_z$. 
We call this map a \emph{state ensemble}; its defining properties are that $\en(z)$ is positive for all $z$, and $\sum_z \tr{\en(z)} = 1$. 
The guessing probability depends on the measurement $\Mo$ that Bob uses, and it is given as
\begin{equation*}
\Pg(\en;\Mo) =  \sum_{z} \tr{\en(z) \Mo(z)} \, .
\end{equation*}
Further, we denote
\begin{equation}\label{eq:def_Pguess}
\Pg(\en) =  \max_{\Mo}  \Pg(\en;\Mo) \, ,
\end{equation}
where the optimization is done over all measurements with outcome set $Z$. 

We are then considering two modifications of the standard state discrimination scenario, where partial classical information concerning the correct label is given either before or after the measurement is performed \cite{BaWeWi08,GoWe10,CaHeTo18,AkKaMa18}. 
The form of the partial information is given as a partitioning $Z=X\cup Y$ of $Z$ into two disjoint subsets. 
By conditioning the state ensemble $\en$ to the occurrence of a label in $X$ or $Y$, we obtain new state ensembles $\en_X$ and $\en_Y$, which we call \emph{subensembles} of $\en$; they are given as
\begin{align*}
\en_X(x) = \tfrac{1}{p(X)} \en(x) \, ,  \qquad \en_Y(y) = \tfrac{1}{p(Y)} \en(y) \, ,
\end{align*}
and their label sets are $X$ and $Y$, respectively.
Here we have denoted $p(X)=\sum_{z\in X} p(z)$ and $p(Y)=\sum_{z\in Y} p(z)$.
We write $\enp=(\en,\{X,Y\})$ for the partitioned state ensemble, i.e., the state ensemble $\en$ with the partitioning of $Z$ into disjoint subsets $X$ and $Y$.

If Alice announces the correct subensemble before Bob  chooses his measurement, we call the task \emph{discrimination with premeasurement information}.
In this case, Bob can choose a measurement $\Ao$ with the outcome set $X$ to discriminate $\en_X$ and a measurement $\Bo$ with the outcome set $Y$ to discriminate $\en_Y$.
At each round of the experiment he measures either $\Ao$ or $\Bo$, depending on Alice's announcement.  
Bob's total guessing probability is
\begin{align}\label{eq:prior}
\Prg(\enp;\Ao,\Bo) = p(X)\Pg(\en_X;\Ao) + p(Y)\Pg(\en_Y;\Bo)
\end{align}
and its maximal value is
\begin{equation}\label{eq:prior-max}
\begin{aligned}
\Prg(\enp) & = \max_{(\Ao,\Bo)\in \OO} \Prg(\enp;\Ao,\Bo) \\
& = p(X)\Pg(\en_X) + p(Y)\Pg(\en_Y) \, . 
\end{aligned}
\end{equation}

In the other variant of the discrimination scenario, Alice announces the correct subensemble only after Bob  has performed his measurement.
Bob has to use a fixed measurement at each round but he can postprocess the obtained measurement outcome according to the additional information. 
We call this task \emph{discrimination with postmeasurement information}.
It has been shown in \cite{CaHeTo18} that now the maximal guessing probability, denoted as $\Ppg(\enp)$, is given by
\begin{align}\label{eq:post-max}
\Ppg(\enp) = \max_{(\Ao,\Bo)\in \cOO} \Prg(\enp;\Ao,\Bo)  \, .
\end{align}
A comparison of \eqref{eq:prior-max} and \eqref{eq:post-max} reveals that the maximal guessing probabilities $\Prg(\enp)$ and $\Ppg(\enp)$ result in optimizing the same mathematical quantity, with the important difference that in the latter the optimization is restricted to compatible pairs of measurements.
From this, we already conclude that if $\Prg(\enp;\Ao,\Bo) > \Ppg(\enp)$ for some partitioned state ensemble $\enp$, then $\Ao$ and $\Bo$ are incompatible.
This conclusion is essentially \cite[Thm.~1]{CaHeTo18}, stated in slightly different words.
In the following, we develop this observation into a necessary and sufficient condition for incompatibility by using the framework of witnesses.

We first notice that, for a partitioned state ensemble $\enp=(\en,\{X,Y\})$ with $\Prg(\enp) > \Ppg(\enp)$, the function
\begin{equation}\label{eq:wit_enp}
\wit_{\enp}(\Ao,\Bo) = \Ppg(\enp) - \Prg(\enp;\Ao,\Bo) 
\end{equation}
is a tight IW for pairs of measurements in $\OO$; we call it the \emph{incompatibility witness associated with $\enp$}. 
In some cases, the exact evaluation of $\Ppg(\enp)$ may be a difficult task, but still by finding a number $\delta$ such that $\Ppg(\enp) \leq \delta <\Prg(\enp)$ one obtains an IW by setting
\begin{equation}\label{eq:wit_enp0}
\wit^\delta_{\enp} (\Ao,\Bo) =\delta - \Prg(\enp;\Ao,\Bo) \, .
\end{equation}
Clearly, we then have $\wit^\delta_{\enp}\preccurlyeq\wit_{\enp}$. 

An important feature of the witnesses arising from partitioned state ensembles is that their physical implementation is straightforward.
Namely, the quantities $\Pg(\en_X;\Ao)$ and $\Pg(\en_Y;\Bo)$ are obtained by performing standard state discrimination experiments, and $\Prg(\enp;\Ao,\Bo)$ is then given via \eqref{eq:prior}.
The constant term $\Ppg(\enp)$ must be calculated analytically or numerically, or at least upper bounded tightly enough. 
It has been shown in \cite{CaHeTo18} that the calculation of $\Ppg(\enp)$ reduces to the evaluation of the standard guessing probability $\Pg(\en')$ of an auxiliary state ensemble $\en'$, and the techniques for calculating the standard guessing probability (see, e.g., \cite{Bae13}) are thereby applicable.


{\em Characterization of incompatibility witnesses.---} The following two theorems are the main results of this paper.

\begin{theorem}\label{teo:main_standard}
For any incompatibility witness $\wit$, there exists a partitioned state ensemble $\enp$ such that the associated incompatibility witness $\wit_{\enp}$ is finer than $\wit$.
Further, if $\wit$ is tight, there exists a partitioned state ensemble $\enp$ such that $\wit$ is detection equivalent to $\wit_{\enp}$.
\end{theorem}

In the case of IWs, the natural choice for the ambient vector space $\V$ containing $\OO$ is the Cartesian product \(\fun{X}\times\fun{Y}\), where $\fun{X}$ is the vector space of all operator valued functions $F\colon X \to \lhs$. 
All linear maps on $\fun{X}\times\fun{Y}$ are expressible in terms of scalar products with elements $(F,G)\in\fun{X}\times\fun{Y}$, so that the basic representation \eqref{eq:witness_main} of witnesses takes the form
\begin{equation}\label{eq:nonstandard_IW}
\wit(\Ao,\Bo) = \delta - \sum_x \tr{F(x)\Ao(x)} - \sum_y \tr{G(y)\Bo(y)}
\end{equation}
for all $(\Ao,\Bo)\in\OO$.
The proof of Thm.~\ref{teo:main_standard} is based on the freedom in the choice of $(F,G)$ and $\delta$.

\begin{proof}[Proof of Thm.~\ref{teo:main_standard}]
Starting from an IW $\wit$ of the general form \eqref{eq:nonstandard_IW}, we similarly define a map $\wit'$ by choosing $F'(x) = \alpha[F(x)-\mu\id]$, $G'(y) = \alpha[G(y)-\mu\id]$ and $\delta' = \alpha(\delta-2\mu d)$, where $d$ is the dimension of the Hilbert space and $\alpha,\mu\in\Rb$ are constants that we determine next.
A direct calculation shows that $\wit' = \alpha\wit$ on $\OO$. 
First, we fix the value of $\mu$ by setting
\begin{equation*}
-\mu = \sum_{x\in X}\no{F(x)} + \sum_{y\in Y}\no{G(y)} \, , 
\end{equation*}
{where \(\no{\cdot}\) denotes the uniform operator norm on \(\lhs\).}  With this choice, all the operators $\en(x) = \abs{\alpha}[F(x)-\mu\id]$ and $\en(y) = \abs{\alpha}[G(y)-\mu\id]$ are positive. 
Secondly, we fix the value of $\alpha$ by setting
\begin{equation*}
\frac{1}{\alpha} = \sum_{x\in X} \tr{F(x) - \mu\id} + \sum_{y\in Y} \tr{G(y) - \mu\id} \,.
\end{equation*}
The right-hand side of this expression is strictly positive, as otherwise $F(x) = G(y) = \mu\id$ for all $x,y$ and so the original IW \eqref{eq:nonstandard_IW} would be constant on $\OO$, which is impossible. 
Thereby, $\alpha > 0$; hence, the map $\wit' = \alpha\wit$ is an IW and $\wit'\approx\wit$. 
Moreover, in this way we have obtained a partitioned state ensemble $\enp=(\en,\{X,Y\})$, for which the witness \(\wit'\) has the form \eqref{eq:wit_enp0}:
$\wit'(\Ao,\Bo) = \delta' - \Prg(\enp;\Ao,\Bo)$.
Since $\wit'$ is an IW and hence satisfies (W1), we must have $\Ppg(\enp)\leq \delta' < \Prg(\enp)$.
Thereby, $\wit' \preccurlyeq\wit_{\enp}$. 
If in addition $\wit$ is tight, then $\delta' = \Ppg(\enp)$, and thus $\wit' = \wit_{\enp}$. 
\end{proof}
 
An important consequence of Thm.~\ref{teo:main_standard} is the following novel operational interpretation for quantum incompatibility.

\begin{theorem}\label{teo:detect_AB}
Two measurements $\Ao$ and $\Bo$ are incompatible if and only if there exists a partitioned state ensemble $\enp$ such that $\Prg(\enp;\Ao,\Bo) > \Ppg(\enp)$.
\end{theorem}

The probability $\Prg(\enp;\Ao,\Bo)$ is assessable by using Alice's classical information, and then performing quantum measurements only on Bob's side. Since no entangled state is shared in the state discrimination protocol, Thm.~\ref{teo:detect_AB} provides a much more practical way to detect incompatibility than schemes based on Bell experiments or steering. In particular, as a fundamental fact, entanglement is not needed to detect incompatibility.

\begin{proof}[Proof of Thm.~\ref{teo:detect_AB}]
The ``if'' statement has already been observed earlier, so here we prove the ``only if'' part. Let us assume that $(\Ao,\Bo)\notin\cOO$. Then, by the usual separation results for compact convex sets \cite[Cor.~11.4.2]{Rock}, there exist $(F,G)\in\fun{X}\times\fun{Y}$ and $\delta\in\Rb$ such that, defining $\wit$ as in \eqref{eq:nonstandard_IW}, we have $\wit(\Ao',\Bo')\geq 0$ for all $(\Ao',\Bo')\in\cOO$ and $\wit(\Ao,\Bo)<0$. 
By Thm.~\ref{teo:main_standard} there exists a partitioned state ensemble $\enp$ such that $\wit\preccurlyeq\wit_{\enp}$. 
It follows that $\wit_{\enp}(\Ao,\Bo) < 0$, i.e., $\Prg(\enp;\Ao,\Bo) > \Ppg(\enp)$.
\end{proof}


{\em Bounding the compatibility region by means of two mutually unbiased bases.---} As we have seen, constructing an IW involves the solution of two convex optimization problems: the evaluation of the maximal guessing probabilities defined in \eqref{eq:prior-max} and \eqref{eq:post-max}. 
In particular, if $\enp$ is a partitioned state ensemble for which the two probabilities differ, whenever the maximum in the right-hand side of \eqref{eq:post-max} admits an analytical computation, one can insert the resulting value of $\Ppg(\enp)$ into \eqref{eq:wit_enp} and thus write the tight IW associated with $\enp$ in an explicit form. 

Interestingly, solving the optimization problem \eqref{eq:post-max} yields even more. Indeed, evaluating a constrained maximum typically requires finding some feasible points where the maximum is attained; if the optimization problem is convex, these points are necessarily located on the relative boundary of the feasible domain. In our specific case, it means that, as a byproduct of solving \eqref{eq:post-max}, we get points lying on the relative boundary $\partial\cOO$ of the convex set $\OO$. Then, by taking convex combinations of these points, we can even have an insight into the set $\cOO$ itself. We thus see that the solution of \eqref{eq:post-max} has a twofold purpose: on the one hand, through the IW constructed in \eqref{eq:wit_enp}, it provides a simple method to detect the incompatibility of many measurement pairs; on the other hand, by using the resulting optimal points, some information on the set of compatible pairs can be inferred.

An interesting special case in which the optimization problems \eqref{eq:prior-max} and \eqref{eq:post-max} admit an analytical solution is when the partitioned state ensemble $\enp$ is made up of two mutually unbiased bases (MUB) of the system Hilbert space $\hh$, or, more generally, smearings of two MUB. Indeed, suppose $\{\phii_h\}_{h\in\{1,\ldots,d\}}$ and $\{\psi_k\}_{k\in\{1,\ldots,d\}}$ is a fixed pair of MUB; then, we can use it to construct a partitioned state ensemble as follows. First, we choose $Z= \{1,\ldots,d\}\times\{\phii,\psi\}$ as the overall label set of the ensemble
\begin{equation}\label{eq:ensemble_MUB}
\en_\bmu (j,\ell) = \frac{1}{2d} \left[\mu_\ell\kb{\ell_j}{\ell_j} + (1-\mu_\ell) \frac{1}{d}\,\id\right] \,,
\end{equation}
where ${\bmu} = (\mu_\phii,\mu_\psi)$ and $\mu_\phii,\mu_\psi\in[1/(1-d),1]$ are real parameters. Next, we partition $Z$ into the subsets $X = \{(1,\phii),\ldots (d,\phii)\}$ and $Y = \{(1,\psi),\ldots (d,\psi)\}$; here, the letters $\phii$ and $\psi$ are just symbols, which are needed to distinguish labels in different subsets. Finally, we set $\enp_\bmu = (\en_\bmu,\{X,Y\})$.

The detailed solution to the optimization problems \eqref{eq:prior-max} and \eqref{eq:post-max} for the partitioned state ensemble $\enp_\bmu$ is provided in Supplemental Material. 
It turns out that the pair of measurements
\begin{equation}\label{eq:def_Alam_Bgam}
\begin{aligned}
\Ao (h,\phii) &= \gamma_\phii\kb{\phii_h}{\phii_h} + (1-\gamma_\phii) \frac{1}{d}\,\id \\\Bo (k,\psi) & = \gamma_\psi\kb{\psi_k}{\psi_k} + (1-\gamma_\psi) \frac{1}{d}\,\id
\end{aligned}
\end{equation}
is a feasible maximum point for a suitable choice of real numbers $\gamma_\phii$ and $\gamma_\psi$, which depend on $\bmu$. 
The next two theorems then follow by our earlier observations.

\begin{theorem}
\label{theo:tightMUBwit}
Let $\bmu = (\mu_\phii,\mu_\psi) \in [1/(1-d),1]\times [1/(1-d),1]$ with $\bmu\neq (0,0)$. Then $\Ppg(\enp_\bmu) < \Prg(\enp_\bmu)$ if and only if $\mu_\phii\mu_\psi\neq 0$ and either $d=2$ or $\max\{\mu_\phii,\mu_\psi\} > 0$. In this case, the tight incompatibility witness associated with the partitioned state ensemble $\enp_\bmu$ is
\begin{equation}\label{eq:MUBwitness}
\begin{aligned}
&\wit_{\enp_\bmu} (\Ao,\Bo) =  \frac{1}{4}\left[\mu_\phii + \mu_\psi + \sqrt{\mu_\phii^2 + \mu_\psi^2 - 2\left(1-\frac{2}{d}\right) \mu_\phii\mu_\psi} \right]  \\
& \quad\ - \frac{1}{2d} \sum_{j=1}^d \left[ \mu_\phii \ip{\phii_j}{\Ao(j,\phii)\phii_j} + \mu_\psi \ip{\psi_j}{\Bo(j,\psi)\psi_j} \right] \,.
\end{aligned}
\end{equation}
Finally, the ensembles $\enp_\bmu$ and $\enp_\bnu$ determine detection equivalent incompatibility witnesses if and only if $\bnu = \alpha\bmu$ for some $\alpha > 0$.
\end{theorem}

By the equivalence statement in the previous theorem, no generality is lost if we express the vector $\bmu$ in terms of a single real parameter $\theta$. Consequently, also the vector $\bgam = (\gamma_\phii,\gamma_\psi)$ parametrizing the optimal measurements \eqref{eq:def_Alam_Bgam} becomes a function of $\theta$. Thus, solving the optimization problem \eqref{eq:post-max} for the present case actually yields a curve in the relative boundary $\partial\cOO$.

\begin{theorem}\label{theo:boundary}
The pair of measurements $(\Ao,\Bo)$ of \eqref{eq:def_Alam_Bgam} lies on the relative boundary $\partial\cOO$ if
\begin{align}\label{eq:boundary}
\bgam = \left(\frac{d-2-d\cos(\theta+\theta_0)}{2(d-1)} \,,\,\frac{d-2-d\cos(\theta-\theta_0)}{2(d-1)} \right)
\end{align}
for $\theta \in[-\theta_0,\theta_0]$ and $\theta_0 = \pi-\arctan\sqrt{d-1}$.
\end{theorem}

When $\theta = 0$, the common value of the two components of \eqref{eq:boundary} is the {\em noise robustness} of the two MUB at hand; it was already derived by different methods in \cite{UoLuMoHe16,DeSkFrBr18}. On the other hand, under the assumption that the two MUB are Fourier conjugate, the portion of the curve \eqref{eq:boundary} with $\gamma_\phii > 0$ and $\gamma_\psi > 0$ was found in \cite{CaHeTo12}.

\begin{figure}[h!]
\centering
\includegraphics[width=0.47\textwidth]{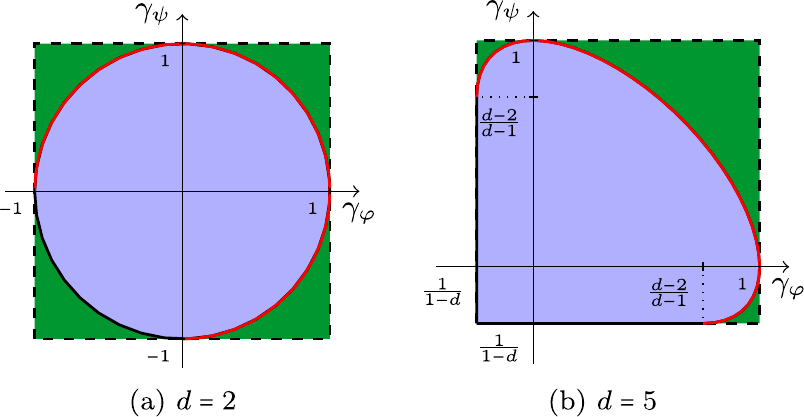}
\caption{The set of $\bgam = (\gamma_\phii,\gamma_\psi)$ for which Eq.~\eqref{eq:def_Alam_Bgam} defines two measurements (green square), and the one for which these measurements are compatible (blue region) for different values of the dimension $d$. The red line is the curve \eqref{eq:boundary}. The case $d=2$ is special, and was already treated in \cite{Busch86}.\label{fig:MUBregion}}
\end{figure}

The operators in \eqref{eq:def_Alam_Bgam} are positive if and only if $\bgam\in [1/(1-d)\,,\,1]\times [1/(1-d)\,,\,1]$. Thus, all pairs of measurements of the form \eqref{eq:def_Alam_Bgam} constitute a square-shaped section of the set $\OO$. Remarkably, the lower-left vertex $(1/(1-d)\,,\,1/(1-d))$ of this square corresponds to a compatible pair of measurements if and only if $d\geq 3$; on the contrary, when $d=2$ the relative boundary is symmetric around $(0,0)$ \cite{Busch86}. Combining these considerations and Thm.~\ref{theo:boundary}, we can give a partial inspection of the two sets $\OO$ and $\cOO$, as shown in Fig.~\ref{fig:MUBregion}.

{\em Discussion.---} The framework of witnesses is an effective tool in the detection of properties described by sets with compact and convex complements.
We have shown that for incompatibility of measurements, witnesses are not only a mathematical tool, but can be implemented in simple discrimination experiments. 
An important feature of this implementation  is that it does not require entanglement.

Our characterization yields a novel operational interpretation of incompatibility: a collection of measurements is incompatible if and only if there is a state discrimination task where premeasurement information is strictly better than postmeasurement information. 

Entanglement witnesses have been used not only to detect entanglement but also to quantify entanglement \cite{EiBrAu07}. 
Further, one can drop the condition (W2) and consider nonlinear witnesses \cite{GuLu06}. 
These and other modifications or generalizations will be an interesting matter of investigation in the case of incompatibility witnesses.


\begin{acknowledgments}

\emph{Acknowledgments.---} The authors thank Anna Jen{\v{c}}ov\'a for the valuable discussion on incompatibility witnesses.
T.~H.~acknowledges financial support from the Academy of Finland via the Centre of Excellence program (Grant No.~312058) as well as Grant No.~287750.
\end{acknowledgments}

\emph{Note added.---} Recently we have been informed about two works that contain related results \cite{UoKrShYuGu18,SkSuCa19}.

\clearpage
\newpage

\onecolumngrid

\renewcommand{\V}{\mathcal{V}}
\renewcommand{\iOO}{\overline{\ca{O}^{\mathrm{com}}_{X,Y}}}
\setcounter{equation}{0}
\setcounter{figure}{0}
\setcounter{definition}{0}
\setcounter{theorem}{0}
\setcounter{proposition}{0}
\setcounter{lemma}{0}
\setcounter{section}{0}
\renewcommand{\theequation}{S\arabic{equation}}
\renewcommand{\thedefinition}{S\arabic{definition}}
\renewcommand{\thetheorem}{S\arabic{theorem}}
\renewcommand{\theproposition}{S\arabic{proposition}}
\renewcommand{\thelemma}{S\arabic{lemma}}

\begin{center}
\textbf{\Large Supplementary Material: Quantum Incompatibility Witnesses}
\end{center}

\bigskip

This supplement is planned as follows. In Sec.~\ref{sec:general_certification}, we provide the general framework for witnesses, recalling the basic related notions from convex analysis and deriving some useful results; in particular, by making a natural assumption upon the convex set to be detected, we fully characterize the equivalence classes of witnesses under the detection equivalence relation. In Sec.~\ref{sec:inco_wit}, these results are applied to incompatibility witnesses; to do it, we describe the convex structure of the set of all pairs of measurements and its subset of all the pairs that are compatible. Finally, in Sec.~\ref{sec:MUB_suppl}, we solve the optimization problem for the guessing probability with postmeasurement information for a state ensemble constructed by means of two mutually unbiased bases; we thus provide the detailed proofs of Thms.~3 and~4 of the main paper.

{\bf Notation:} In this supplement, the numberings of equations, theorems etc.~are preceeded by the letter `S' (e.g.: Eq.~(S1), Thm.~S1 etc.). When we refer to results in the main paper, we simply drop the `S' (Eq.~(1), Thm.~1 etc.).

\section{Witnesses for general convex sets}\label{sec:general_certification}

\subsection{Preliminaries from convex analysis}\label{subsec:convex_preliminaries}

In the following, we will need some standard terminology and notations from convex analysis. We refer to \cite{SRock} for further details.

Suppose $\V$ is a finite dimensional, real and normed linear space. If $\V_0\subseteq\V$ is a linear subspace, we denote by $\V_0^*$ the dual linear space of $\V_0$, and $\pair{v}{v^*}$ the canonical pairing between an element $v\in\V_0$ and a dual vector $v^*\in\V_0^*$.

We recall that
\begin{enumerate}[-]
\item an {\em affine} [respectively, {\em convex}] {\em combination} of elements $v_1,\ldots,v_n\in\V$ is any linear combination $\lam_1 v_1 + \ldots + \lam_n v_n$ such that $\lam_1 + \ldots + \lam_n = 1$ [resp., such that $\lam_1 + \ldots + \lam_n = 1$ and $\lam_i\in [0,1]$ for all $i=1,\ldots,n$];
\item an {\em affine} [resp., {\em convex}] {\em set} is any subset $\Sscr\subseteq\V$ such that all affine [resp., convex] combinations of elements of $\Sscr$ are still contained in $\Sscr$;
\item an {\em affine} [resp., {\em c-affine}] {\em  map} on an affine [resp., convex] set $\Sscr$ is a function $\wit:\Sscr\to\Rb$ such that
$$
\wit(\lam_1 v_1 + \ldots + \lam_n v_n) = \lam_1\wit(v_1) + \ldots + \lam_n \wit(v_n)
$$
for all $v_1,\ldots,v_n\in\Sscr$ and any affine [resp., convex] combination $\lam_1 v_1 + \ldots + \lam_n v_n$.
\end{enumerate}
If $\Mscr$ is an affine set, there is a unique vector subspace $\V(\Mscr)\subseteq\V$ such that $\Mscr = v_0 + \V(\Mscr)$ for some (hence for any) $v_0\in\Mscr$. Moreover, if $\wit:\Mscr\to\Rb$ is an affine map and $v_0\in\Mscr$ is fixed, there exist unique $v_\wit^*\in\V(\Mscr)^*$ and $\delta_\wit\in\Rb$ such that $\wit(v_0 + v) = \delta_\wit - \pair{v}{v_\wit^*}$ for all $v\in\V(\Mscr)$. In particular, given an affine map $\Wit:\V\to\Rb$, there exist unique $v^*\in\V^*$ and $\delta\in\Rb$ such that
\equasi{eq:general_functional}{
\Wit(v) = \delta-\pair{v}{v^*} =: \Wit[v^*,\delta](v) \qquad \forall v\in\V \,.
}
By picking any $v^*\in\V^*$ whose restriction to $\V(\Mscr)$ coincides with $v_\wit^*$ and choosing $\delta = \delta_\wit - \pair{v_0}{v^*}$, we see that the previous affine map $\wit:\Mscr\to\Rb$ extends to the affine map $\Wit[v^*,\delta]:\V\to\Rb$. Clearly, this extension is not unique unless $\V(\Mscr) = \V$.

Now, suppose $\Cscr\subseteq\V$ is a convex set. Then, the {\em affine hull} of $\Cscr$ is the smallest affine set $\Aff{\Cscr}$ containing $\Cscr$; equivalently, it is the set of all affine combinations of elements of $\Cscr$. We abbreviate $\V(\Cscr) = \V(\Aff{\Cscr})$; further, we introduce the following {\em annihilator subspace} of $\V(\Cscr)$ in $\V^*$\,:
$$
\V(\Cscr)^\perp = \{v^*\in\V^*\mid \pair{v}{v^*} = 0 \ \ \forall v\in\V(\Cscr)\} \,.
$$
Any c-affine map $\wit:\Cscr\to\Rb$ uniquely extends to an affine map $\tilde{\wit}:\Aff{\Cscr}\to\Rb$; as we have already seen, such a map $\tilde{\wit}$ can be further extended to an affine map $\Wit:\V\to\Rb$, the latter extension being in general not unique. Actually, the first assertion is a particular case of the following more general result \cite{SVesely}.

\begin{proposition}\label{prop:extension}
Suppose $\Cscr_1$ and $\Cscr_2$ are convex sets, and let $\phi : \Cscr_1\to \Cscr_2$ be a map such that $\phi(\lam_1 v_1 + \ldots + \lam_n v_n) = \lam_1\phi(v_1) + \ldots + \lam_n\phi(v_n)$ for all convex combinations of elements $v_1,\ldots,v_n\in\Cscr_1$. Then, there exists a unique map $\tilde{\phi} : \Aff{\Cscr_1}\to \Aff{\Cscr_2}$ such that the restriction $\tilde{\phi}\big|_{\Cscr_1}$ coincides with $\phi$ and $\tilde{\phi}(\lam_1 v_1 + \ldots + \lam_n v_n) = \lam_1\tilde{\phi}(v_1) + \ldots + \lam_n\tilde{\phi}(v_n)$ for all affine combinations of elements $v_1,\ldots,v_n\in\Aff{\Cscr_1}$.
\end{proposition}
\begin{proof}
Any element of the affine hull $\Aff{\Cscr_i}$ is an affine combination of elements of $C_i$. Then, if $x = \lam_1 v_1 + \ldots + \lam_n v_n \in\Aff{\Cscr_1}$ is any affine combination of $v_1,\ldots,v_n\in\Cscr_1$, we define
$$
\tilde{\phi}(x) = \lam_1 \phi(v_1) + \ldots + \lam_n \phi(v_n) \,.
$$
First of all, we claim that this definition is independent of the chosen representation of $x$. Indeed, suppose $x = \mu_1 w_1 + \ldots + \mu_m w_m$ for some other $w_1,\ldots,v_m\in\Cscr_1$ and $\mu_1,\ldots,\mu_m\in\Rb$ with $\mu_1 + \ldots + \mu_m = 1$. We set
$$
\sigma = \sum_{i \mid \lam_i > 0} \lam_i - \sum_{j \mid \mu_j < 0} \mu_j = \sum_{j \mid \mu_j > 0} \mu_j - \sum_{i \mid \lam_i < 0} \lam_i \,,
$$
and observe that $\sigma > 0$ since at least one among $\lam_1,\ldots,\lam_n$ is necessarily strictly positive. Then, as
$$
\sum_{i \mid \lam_i > 0} \frac{\lam_i}{\sigma}\, v_i + \sum_{j \mid \mu_j < 0} \frac{-\mu_j}{\sigma}\, w_j = \sum_{j \mid \mu_j > 0} \frac{\mu_j}{\sigma}\, w_j + \sum_{i \mid \lam_i < 0} \frac{-\lam_i}{\sigma}\, v_i \,,
$$
and the two sides of the latter equation are convex combinations of elements of $\Cscr_1$, it follows that
$$
\sum_{i=1}^n \lam_i \phi(v_i) - \sum_{j=1}^m \mu_j \phi(w_j)  = \sigma \left[\phi\left( \sum_{i \mid \lam_i > 0} \frac{\lam_i}{\sigma}\, v_i + \sum_{j \mid \mu_j < 0} \frac{-\mu_j}{\sigma}\, w_j \right) - \phi\left( \sum_{j \mid \mu_j > 0} \frac{\mu_j}{\sigma}\, w_j + \sum_{i \mid \lam_i < 0} \frac{-\lam_i}{\sigma}\, v_i \right)\right] = 0 \,,
$$
thus showing that $\tilde{\phi}$ is well defined.\\
Secondly, we prove that $\tilde{\phi}$ is an affine map, that is
\equasi{eq:def_tilde_phi}{\tag{$\circ$}
\tilde{\phi}\left(\sum_{i=1}^n \lam_i u_i\right) = \sum_{i=1}^n \lam_i \tilde{\phi}(u_i) \quad \text{for all} \quad u_1,\ldots, u_n\in\Aff{C_1} \quad \text{and} \quad \lam_1,\ldots,\lam_n\in\Rb \quad \text{with} \quad \sum_{i=1}^n \lam_i = 1 \,.
}
To do it, write each $u_i$ as an affine combination $u_i = \mu_{i,1}v_{i,1} + \ldots + \mu_{i,m_i}v_{i,m_i}$ of elements $v_{i,1},\ldots, v_{i,m_i}\in\Cscr_1$. Then,
$$
\sum_{i=1}^n \lam_i u_i = \sum_{i=1}^n \sum_{j=1}^{m_i} \lam_i \mu_{i,j} v_{i,j} \,,
$$
where the right hand side is an affine combination of the elements $v_{i,j}\in\Cscr_1$. Hence,
$$
\tilde{\phi}\left(\sum_{i=1}^n \lam_i u_i\right) = \sum_{i=1}^n \sum_{j=1}^{m_i} \lam_i \mu_{i,j} \phi(v_{i,j}) \qquad \text{and} \qquad \tilde{\phi}(u_i) = \sum_{j=1}^{m_i} \mu_{i,j} \phi(v_{i,j})
$$
by definition of $\tilde{\phi}$. This proves \eqref{eq:def_tilde_phi}.\\
Finally, the equality $\tilde{\phi}\big|_{\Cscr_1} = \phi$ also follows by the very definition of $\tilde{\phi}$.
\end{proof}

The {\em relative interior} $\ri{\Cscr}$ of the convex set $\Cscr$ is the set of all its interior points with respect to the relative topology of $\Aff{\Cscr}$. The {\em relative boundary} of $\Cscr$ is the set-theoretic difference $\partial\Cscr = \Cscr\setminus\ri{\Cscr}$. An element $z\in\Cscr$ is an {\em extreme point} of $\Cscr$ if the equality $z=\lam x + (1-\lam) y$ with $x,y\in\Cscr$ and $\lam\in(0,1)$ implies $x=y=z$.

If $\Cscr$ is a compact convex set, its {\em support function} is the map
$$
\delta_\Cscr : \V^*\to\Rb\,, \qquad \qquad \delta_\Cscr(v^*) = \max\{\pair{x}{v^*}\mid x\in\Cscr\} \,.
$$
Clearly, if $\Cscr_0\subseteq\Cscr$ is a compact convex subset, then $\delta_{\Cscr_0}(v^*) \leq \delta_\Cscr(v^*)$ for all $v^*\in V^*$. The support function satisfies $\delta_\Cscr(\alpha v^*) = \alpha \delta_\Cscr(v^*)$ for all $\alpha > 0$; moreover, if $v^*-u^*\in\V(\Cscr)^\perp$, then $\delta_\Cscr(v^*) - \delta_\Cscr(u^*) = \pair{v_0}{v^*-u^*}$ for some (hence for all) $v_0\in\Aff{\Cscr}$.

\subsection{Detecting a convex property}

We are interested in a set of objects $\Cscr$, in which we consider the subset $\Cscr_0\subset\Cscr$ of all the objects sharing some given property. We assume $\emptyset\neq\Cscr_0\neq\Cscr$, and we denote by $\overline{\Cscr}_0 = \Cscr\setminus\Cscr_0$ the subset of all the objects which {\em do not} possess the property at hand; we want to find simple sufficient conditions guaranteeing that an object $x\in\Cscr$ actually belongs to $\overline{\Cscr}_0$.

In the following, we always suppose that both $\Cscr$ and $\Cscr_0$ are convex and compact subsets of a finite dimensional, real and normed linear space $\V$. Then, the simplest conditions involve some specific c-affine map $\wit:\Cscr\to\Rb$, related to both $\Cscr$ and $\Cscr_0$, and the value that $\wit$ takes at $x\in\Cscr$.
\begin{definition}\label{def:Cc_witness}
A {\em $\overline{\Cscr}_0$-witness} on the convex set $\Cscr$ is a c-affine map $\wit:\Cscr\to\Rb$ such that
\begin{enumerate}[(i)]
\item $\wit(y)\geq 0$ for all $y\in\Cscr_0$;\label{it:1_def:Cc_witness}
\item $\wit(x) < 0$ for some $x\in\overline{\Cscr}_0$.\label{it:2_def:Cc_witness}
\end{enumerate}
A $\overline{\Cscr}_0$-witness is {\em tight} if it satisfies the further condition
\begin{enumerate}[(i)]\setcounter{enumi}{2}
\item $\wit(z) = 0$ for some $z\in\Cscr_0$.\label{it:3_def:Cc_witness}
\end{enumerate}
\end{definition}
If $\wit$ is a $\overline{\Cscr}_0$-witness, the inequality $\wit(x)<0$ entails that $x\in\Cscr$ does not possess the property $\Cscr_0$. The set $\De{\wit} = \{x\in\Cscr\mid\wit(x)<0\}$ is thus the subset of all the objects of $\overline{\Cscr}_0$ which are {\em detected} by $\wit$. Note that there always exists a tight $\overline{\Cscr}_0$-witness $\wit'$ detecting at least as many objects of $\overline{\Cscr}_0$ as $\wit$. Namely, it is enough to set $\wit'(x) = \wit(x)-\min\{\wit(y)\mid y\in\Cscr_0\}$. In general, whenever two $\overline{\Cscr}_0$-witnesses $\wit$ and $\wit'$ are such that $\De{\wit}\subseteq\De{\wit'}$, we say that $\wit'$ is {\em finer} than $\wit$, and in this case we write $\wit\preccurlyeq\wit'$.  Moreover, if $\De{\wit} = \De{\wit'}$, we say that $\wit$ and $\wit'$ are {\em detection equivalent} and write $\wit\approx\wit'$.

\subsection{Structure of $\overline{\Cscr}_0$-witnesses}

\begin{proposition}\label{prop:cond_delta}
A c-affine map $\wit:\Cscr\to\Rb$ is a $\overline{\Cscr}_0$-witness if and only if there exist $v^*\in\V^*$ and $\delta\in\Rb$ such that
\begin{enumerate}[(i)]
\item $\delta_{\Cscr_0}(v^*) \leq \delta < \delta_\Cscr(v^*)$;\label{it:1_cond_delta}
\item $\wit(x) = \delta - \pair{x}{v^*}$ for all $x\in\Cscr$.\label{it:2_cond_delta}
\end{enumerate}
In \eqref{it:1_cond_delta}, the equality is attained if and only if $\wit$ is tight.
\end{proposition}
\begin{proof}
Any c-affine map $\wit$ on $\Cscr$ extends to an affine map $\Wit = \Wit[v^*,\delta]$ on $\V$, and then $\wit(x) = \delta - \pair{x}{v^*} \ \ \forall x\in\Cscr$ by \eqref{eq:general_functional}. In this case, $\wit$ is a $\overline{\Cscr}_0$-witness if and only if, for all $y\in\Cscr_0$ and some $x\in\overline{\Cscr}_0$,
$$
\Wit[v^*,\delta](y) \geq 0 > \Wit[v^*,\delta](x) \quad \Leftrightarrow \quad \pair{y}{v^*} \leq \delta < \pair{x}{v^*} \,.
$$
This is equivalent to $\delta_{\Cscr_0}(v^*) \leq \delta < \delta_\Cscr(v^*)$, where the equality is attained if and only if $\pair{z}{v^*} = \delta$ for some $z\in\Cscr_0$, that is, $\wit(z) = 0$.
\end{proof}

If $\V$ is an Euclidean space and $v^*$ is given by the scalar product $\pair{v}{v^*} = \hat{e}\cdot v\ \ \forall v\in\V$ for some unit vector $\hat{e}\in\V$ not belonging to $\V(\Cscr_0)^\perp$, the nonnegative gap \(\delta_\Cscr(v^*) -\delta_{\Cscr_0}(v^*)\) is the distance between the affine hyperplane of $\V$ which is orthogonal to $\hat{e}$, touches $\Cscr$ at its relative boundary and has $\Cscr$ on its side opposite to $\hat{e}$, and the analogous hyperplane touching the relative boundary of $\Cscr_0$.

\subsection{Tight witnesses}

The following proposition establishes the natural connection between tight $\overline{\Cscr}_0$-witnesses and the extremality property for points of the set $\Cscr_0$. Although it is an easy consequence of more general and standard results (see e.g.~\cite[Thm.~32.1 and Cor.~32.3.1]{SRock}), for the reader's convenience we provide a simple proof adapted to the present case.
\begin{proposition}\label{prop:tightness}
Suppose $\wit$ is a tight $\overline{\Cscr}_0$-witness on $\Cscr$. Then, the following facts hold.
\begin{enumerate}[(a)]
\item $\wit(z_0) = 0$ for some extreme point $z_0$ of $\Cscr_0$.\label{it:1_prop_tightness}
\item If $\V(\Cscr_0) = \V(\Cscr)$ and $z\in\Cscr_0$, the equality $\wit(z) = 0$ implies that $z\in\partial\Cscr_0$.\label{it:2_prop_tightness}
\end{enumerate}
\end{proposition}

The proof relies on the following lemma, that is sometimes useful by itself.

\begin{lemma}\label{lem:delta}
For $v^*\in\V^*$, the following facts hold.
\begin{enumerate}[(a)]
\item $\delta_{\Cscr_0}(v^*) = \pair{z_0}{v^*}$ for some extreme point $z_0$ of $\Cscr_0$.\label{it:1_lem_delta}
\item Suppose $v^*\in\V^*\setminus\V(\Cscr_0)^\perp$. Then, if $z\in\Cscr_0$, the equality $\pair{z}{v^*} = \delta_{\Cscr_0}(v^*)$ implies that $z\in\partial\Cscr_0$.\label{it:2_lem_delta}
\end{enumerate}
\end{lemma}
\begin{proof}
\eqref{it:1_lem_delta} The set $Z = \{z\in\Cscr_0\mid\pair{z}{v^*} = \delta_{\Cscr_0}(v^*)\}$ is a nonempty closed and convex subset of $\Cscr_0$. By Krein-Milman theorem \cite[Cor.~18.5.1]{SRock}, $Z$ has some extreme point $z_0$. We claim that $z_0$ is extreme also for $\Cscr_0$. Indeed, suppose that $z_0=\lam y_1 + (1-\lam) y_2$ with $y_1,y_2\in\Cscr_0$ and $\lam\in(0,1)$. The conditions $\delta_{\Cscr_0}(v^*) = \pair{z_0}{v^*} = \lam \pair{y_1}{v^*} + (1-\lam) \pair{y_2}{v^*}$ and $\pair{y_i}{v^*}\leq \delta_{\Cscr_0}(v^*)$ for $i=1,2$ then imply that $y_1,y_2\in Z$, hence $y_1=y_2=z_0$.
 
\eqref{it:2_lem_delta} Suppose by contradiction that $z\in\ri{\Cscr_0}$. Then, for any $v\in\V(\Cscr_0)$ with $\pair{v}{v^*}\neq 0$, there exists $\varepsilon\in\Rb$ such that $z+\varepsilon v\in\Cscr_0$ and $\varepsilon \pair{v}{v^*} > 0$. It follows that $\pair{z+\varepsilon v}{v^*} > \delta_{\Cscr_0}(v^*)$, which is impossible.
\end{proof}

\begin{proof}[Proof of Prop.~\ref{prop:tightness}]
By Prop.~\ref{prop:cond_delta}, for some $v^*\in\V^*$ we have $\wit(x) = \delta_{\Cscr_0}(v^*) - \pair{x}{v^*}$ for all $x\in\Cscr$, and $\delta_{\Cscr_0}(v^*) < \delta_\Cscr(v^*)$; in particular, $v^*$ is not constant on $\Cscr$, hence $v^*\in\V^*\setminus\V(\Cscr)^\perp$. Since $\wit(z) = 0$ is then equivalent to $\pair{z}{v^*} = \delta_{\Cscr_0}(v^*)$, the two claims follow by the analogous statements of Lemma \ref{lem:delta}.
\end{proof}

\subsection{Detection equivalent witnesses}

If the subset $\Cscr_0$ is sufficiently large in $\Cscr$, we have the following characterization of detection equivalence.
\begin{proposition}\label{prop:general_equivalence}
Suppose $\Cscr_0\cap\ri{\Cscr}\neq\emptyset$. Then, for two $\overline{\Cscr}_0$-witnesses $\wit_1$ and $\wit_2$ on $\Cscr$, the following facts are equivalent.
\begin{enumerate}[(a)]
\item $\wit_1\approx\wit_2$.\label{it:1_general_equivalence}
\item $\wit_1 = \alpha\wit_2$ for some $\alpha > 0$.\label{it:2_general_equivalence}
\item If $v_i^*\in\V^*$ and $\delta_i\in\Rb$ are such that $\wit_i(x) = \delta_i - \pair{x}{v_i^*}$ for all $x\in\Cscr$ and $i=1,2$, then there exists $\alpha>0$ such that\label{it:3_general_equivalence}
\begin{enumerate}[(i)]
\item $v^*_1 - \alpha v^*_2\in\V(\Cscr)^\perp$;
\item $\delta_1 - \alpha\delta_2 = \pair{v_0}{v^*_1 - \alpha v^*_2}$ for some (hence for all) $v_0\in\Aff{\Cscr}$.
\end{enumerate}
\end{enumerate}
\end{proposition}
Note that, if the two $\overline{\Cscr}_0$-witnesses in the above proposition are tight, then in statement \eqref{it:3_general_equivalence} we have $\delta_i = \delta_{\Cscr_0}(v^*_i)$ by Prop.~\ref{prop:cond_delta}, hence item (\ref{it:3_general_equivalence}.ii) is automatically implied by (\ref{it:3_general_equivalence}.i) and the properties of the support function $\delta_{\Cscr_0}$ recalled at the end of Sec.~\ref{subsec:convex_preliminaries}.

Prop.~\ref{prop:general_equivalence} exhibits the degree of freedom one has in choosing detection equivalent $\overline{\Cscr}_0$-witnesses: namely, if we have a $\overline{\Cscr}_0$-witness $\wit(x) = \delta - \pair{x}{v^*} \ \ \forall x\in\Cscr$, we can turn it into an equivalent one by: (1) replacing the dual vector $v^*\in\V^*$ with $v^{\prime\ast} = \alpha(v^* + u^*)$, where $u^* \in \V(\Cscr)^\perp$ and $\alpha >0$; (2) redefining the constant $\delta$ according to Prop.~\ref{prop:general_equivalence}.(\ref{it:3_general_equivalence}.ii). The smaller is the linear space $\V(\Cscr)$ inside $\V$, the larger is the freedom in the choice of the dual vector $v^*$. This freedom is the crux of the proof of Thm.~1  in the main paper, and therefore it is at the heart of the interpretation of incompatibility witnesses as postmeasurement discrimination problems.

\begin{proof}[Proof of Prop.~\ref{prop:general_equivalence}]
\eqref{it:1_general_equivalence} $\Rightarrow$ \eqref{it:2_general_equivalence}\,: Assuming statement \eqref{it:1_general_equivalence}, we preliminarly show that the two nonempty sets $Z_i = \{z\in\Cscr\mid\wit_i(z) = 0\}$ ($i=1,2$) actually coincide. Indeed, suppose by contradiction that $Z_1\neq Z_2$. We can assume with no restriction that there is some $z\in\Cscr$ such that $\wit_1(z) = 0$ and $\wit_2(z) \neq 0$, hence $\wit_2(z) > 0$ since $\De{\wit_1} = \De{\wit_2} =: \ca{D}$. Picking any $x\in\ca{D}$, we have $\wit_1(\lam x + (1-\lam)z) = \lam\wit_1(x) < 0$, or, equivalently, $\wit_2(\lam x + (1-\lam)z) < 0$ for all $\lam\in (0,1)$, that contradicts continuity of the mapping $\lam\mapsto\wit_2(\lam x + (1-\lam)z)$ at $\lam = 0$.\\
We next claim that there is $z_0\in\ri{\Cscr}$ such that $\wit_i(z_0) = 0$ for all $i=1,2$. To this aim, let $y\in\Cscr_0\cap\ri{\Cscr}$. Then, for all $i=1,2$, either $\wit_i(y) = 0$ and we are done, or $\wit_i(y) > 0$. In the latter case, again by a continuity argument, for any $x\in\ca{D}$ there is some $\lam\in (0,1)$ such that $\wit_1(\lam x + (1-\lam)y) = \lam\wit_1(x) + (1-\lam)\wit_1(y) = 0$. Setting $z_0 = \lam x + (1-\lam)y$, we thus see that $z_0\in Z_1=Z_2$, and $z_0\in\ri{\Cscr}$ by \cite[Thm.~6.1]{SRock}.\\
For $i=1,2$, let $\tilde{\wit}_i$ be the extension of $\wit_i$ to an affine map on $\Aff{\Cscr}$. Then, the mapping $\V(\Cscr)\ni v\mapsto\tilde{\wit}_1(z_0+v)\in\Rb$ is linear and nonzero, hence there exists a linear basis $\{v_1,\ldots,v_m\}$ of $\V(\Cscr)$ such that $\tilde{\wit}_1(z_0+v_k) = 0$ if $k=1,\ldots,m-1$, and $\tilde{\wit}_1(z_0+v_m) < 0$. By possibly replacing all the $v_k$'s with $\mu v_k$ for some $\mu\in (0,1)$, we can assume that $z_k:= z_0+v_k\in\ri{\Cscr}$ for all $k=1,\ldots,m$, and so $\wit_1(z_k) = 0$ if $k\leq m-1$, and $\wit_1(z_m) < 0$. Hence, also $\wit_2(z_k) = 0$ if $k\leq m-1$, and $\wit_2(z_m) < 0$. It follows that
$$
\left[\tilde{\wit}_1 - \frac{\wit_1(z_m)}{\wit_2(z_m)} \tilde{\wit}_2\right](z_k) = 0 \qquad \forall k=0,\ldots,m \,,
$$
which implies $\tilde{\wit}_1 - (\wit_1(z_m) / \wit_2(z_m))\, \tilde{\wit}_2 = 0$ because any element of $\Aff{\Cscr}$ is an affine combination of $z_0,\ldots,z_m$. This yields statement \eqref{it:2_general_equivalence}.

\eqref{it:2_general_equivalence} $\Rightarrow$ \eqref{it:1_general_equivalence}\,: The implication is clear.

\eqref{it:2_general_equivalence} $\Leftrightarrow$ \eqref{it:3_general_equivalence}\,: Suppose $v_i^*\in\V^*$ and $\delta_i\in\Rb$ are as in statement \eqref{it:3_general_equivalence}, and extend the c-affine map $\wit_i$ to an affine map $\tilde{\wit}_i$ on $\Aff{\Cscr}$ by means of the relation $\tilde{\wit}_i(x) = \delta_i - \pair{x}{v^*_i} \ \ \forall x\in\Aff{\Cscr}$. Then, since affine extensions are unique, statement \eqref{it:2_general_equivalence} is equivalent to $\tilde{\wit}_1 = \alpha\tilde{\wit}_2$. Picking any $v_0\in\Aff{\Cscr}$, this is in turn equivalent to
\begin{align*}
0 & = (\tilde{\wit}_i - \alpha \tilde{\wit}_i) (v_0+v) = [(\delta_1 - \alpha\delta_2) - \pair{v_0}{v^*_1 - \alpha v^*_2}] - \pair{v}{v^*_1 - \alpha v^*_2} \qquad \forall v\in\V(\Cscr) \,,
\end{align*}
that is the same as statement \eqref{it:3_general_equivalence}.
\end{proof}

\section{Incompatibility witnesses}\label{sec:inco_wit}

We fix a finite dimensional complex Hilbert space $\hh$, with $\dim\hh = d$. We denote by $\lhs$ the real linear space of all selfadjoint operators on $\hh$, endowed with the uniform operator norm $\no{\cdot}$. We write $\id$ for the identity operator. If $Z$ is a set, we let $\abs{Z}$ be its cardinality. A {\em measurement} with outcomes in a finite set $Z$ is any map $\Mo : Z\to\lhs$ such that $\Mo(z)\geq 0$ for all $z\in Z$ and $\sum_{z\in Z}\Mo(z) = \id$. The {\em uniform measurement} with outcomes in $Z$ is given by $\Uo_Z(z) = \id/|Z|$ for all $z\in Z$.

All measurements with outcomes in $Z$ constitute a closed and bounded convex subset $\obs{Z}$ in the real linear space of all operator valued functions $H:Z\to\lhs$. We denote by $\fun{Z}$ the latter linear space of functions, and we regard it as a normed space with the sup-norm $\no{H}_\infty = \max\{\no{H(z)}\mid z\in Z\}$; the dimension of $\fun{Z}$ is $d^2\abs{Z}$. For any $A\in\lhs$, we define the affine set $\funo{A}{Z} = \{H\in\fun{Z}\mid\sum_{z\in Z} H(z) = A\}$; the inclusion $\obs{Z}\subset\funo{\id}{Z}$ is clear.

If $X$ and $Y$ are finite sets, two measurements $\Ao\in\obs{X}$ and $\Bo\in\obs{Y}$ are {\em compatible} if there exists a third measurement $\Mo\in\obs{X\times Y}$ such that $\Ao$ and $\Bo$ are the {\em margins} of $\Mo$, that is,
\equasi{eq:defC}{
\Mo_{X}(x) := \sum_{y\in Y} \Mo(x,y) = \Ao(x) \qquad\text{and}\qquad \Mo_{Y}(y) := \sum_{x\in X} \Mo(x,y) = \Bo(y)
}
for all $x$ and $y$. In this case, we say that $\Mo$ is a {\em joint measurement} of $\Ao$ and $\Bo$.
We denote by $\OO = \obs{X}\times\obs{Y}$ the set of all pairs of measurements on $X$ and $Y$, and by $\cOO$ the subset of all pairs made up of compatible measurements. If $(\Ao,\Bo)\in\iOO=\OO\setminus\cOO$, the two measurements $\Ao$ and $\Bo$ are {\em incompatible}.

As it is well known, unless $\hh = \Cb$ or $\min\{|X|,|Y|\} = 1$, the inclusion $\cOO\subset\OO$ is strict. So, in the following we will always assume $d\geq 2$ and $\min\{|X|,|Y|\}\geq 2$ to avoid degeneracies.

The sets $\OO$ and $\cOO$ are convex and compact in the direct product linear space $\V = \fun{X}\times\fun{Y}$; here, as the norm of $\V$ we choose the $\ell_\infty$-norm $\no{(F,G)}_\infty = \max\{\no{F}_\infty,\no{G}_\infty\}$. Indeed, only the compactness of $\cOO$ needs to be checked; it follows by the compactness of $\obs{X\times Y}$ and the continuity of the mapping $\obs{X\times Y}\ni\Mo\mapsto (\Mo_{X},\Mo_{Y}) \in\obs{X}\times\obs{Y}$.

The next proposition gives some further insight into the convex structure of the sets $\OO$ and $\cOO$.

\begin{proposition}\label{prop:convex_structure}
The following properties hold.
\begin{enumerate}[(a)]
\item $\Aff{\cOO} = \Aff{\OO} = \funo{\id}{X}\times\funo{\id}{Y}$.\label{it:a_convex_structure}
\item $\V(\cOO) = \V(\OO) = \funo{0}{X}\times\funo{0}{Y}$.\label{it:b_convex_structure}
\item $\ri{\cOO}\subset\ri{\OO}$ and $(\Uo_X,\Uo_Y)\in\ri{\cOO}$.\label{it:c_convex_structure}
\end{enumerate}
\end{proposition}
\begin{proof}
Clearly,
\equasi{eq:inclusions}{\tag{$\ast$}
\Aff{\cOO} \subseteq \Aff{\OO} \subseteq \funo{\id}{X}\times\funo{\id}{Y} = (\Uo_X,\Uo_Y) + \funo{0}{X}\times\funo{0}{Y} \,.
}
In order to prove that the previous inclusions actually are equalities (item \eqref{it:a_convex_structure}), it is enough to show that
\equasi{eq:neighborhood}{\tag{$\ast\ast$}
(\Uo_X,\Uo_Y) + \varepsilon B \subseteq \cOO \quad\text{for}\quad \varepsilon = \frac{1}{2}\, \min \left\{ \frac{1}{\abs{X}} , \frac{1}{\abs{Y}} \right\} \quad\text{and}\quad B = \{(F,G)\in\funo{0}{X}\times\funo{0}{Y} \mid \no{(F,G)}_\infty < 1\}\,;
}
here, $B$ is the open unit ball in $\funo{0}{X}\times\funo{0}{Y}$. Indeed, if \eqref{eq:neighborhood} holds, then
$$
(\Uo_X,\Uo_Y) + \funo{0}{X}\times\funo{0}{Y} = \Aff{(\Uo_X,\Uo_Y) + \varepsilon B} \subseteq \Aff{\cOO} \,.
$$
Now, \eqref{eq:neighborhood} immediately follows, since for $(F,G)\in B$, the formula
$$
\Mo(x,y) = \Uo_{X\times Y}(x,y) + \varepsilon\left[\frac{1}{\abs{Y}}\, F(x) + \frac{1}{\abs{X}}\, G(y)\right]
$$
defines an element $\Mo\in\obs{X\times Y}$ such that $(\Mo_{X},\Mo_{Y}) = (\Uo_X,\Uo_Y) + \varepsilon (F,G)$.
Having proven that in \eqref{eq:inclusions} all inclusions actually are equalities, item \eqref{it:b_convex_structure} is obvious, while item \eqref{it:c_convex_structure} follows from \eqref{eq:neighborhood} and \cite[Cor.~6.5.2]{SRock}.
\end{proof}

The dual space $\V^*$ of $\V = \fun{X}\times\fun{Y}$ can be identified with $\V$ itself by means of the pairing
\equasi{eq:iden_V*}{
\pair{(F_2,G_2)}{(F_1,G_1)} = \sum_{x\in X} \tr{F_1(x)F_2(x)} + \sum_{y\in Y} \tr{G_1(y)G_2(y)} \,.
}
With this identification, Prop.~\ref{prop:convex_structure}.\eqref{it:b_convex_structure} and a simple dimension counting lead to the equalities
\equasi{eq:annihilator_obs}{
\begin{aligned}
\V(\cOO)^\perp & = \V(\OO)^\perp = \{(F,G)\in\fun{X}\times\fun{Y} \mid F(x_1) = F(x_2) \\
& \qquad\qquad\qquad\qquad\qquad \text{and} \quad G(y_1) = G(y_2) \quad \forall x_1,x_2\in X,\, y_1,y_2\in Y\} \,.
\end{aligned}
}

We recall from the main paper that any $\iOO$-witness on the convex set $\OO$ is an {\em incompatibility witness} (IW). Thus, by Prop.~\ref{prop:cond_delta} and \eqref{eq:iden_V*}, any IW is of the form
\equano{
\wit(\Ao,\Bo) = \delta - \pair{(\Ao,\Bo)}{(F,G)} \qquad \forall (\Ao,\Bo)\in\OO
}
for some $(F,G)\in\fun{X}\times\fun{Y}$ and $\delta\in [\delta_{\cOO}(F,G)\,,\,\delta_{\OO}(F,G))$. In particular, $\wit$ is tight if in the above formula $\delta = \delta_{\cOO}(F,G)$.

Combining Props.~\ref{prop:tightness} and \ref{prop:convex_structure}.\eqref{it:b_convex_structure} immediately proves the following connection between tight IWs and the relative boundary of the set $\cOO$.

\begin{proposition}
Suppose $\wit$ is a tight IW on the set $\OO$ and $(\Ao,\Bo)\in\cOO$. Then, the equality $\wit(\Ao,\Bo) = 0$ implies that $(\Ao,\Bo)\in\partial\cOO$. Moreover, there always exists some extreme point $(\Ao_0,\Bo_0)$ of $\cOO$ such that $\wit(\Ao_0,\Bo_0) = 0$.
\end{proposition}

As a consequence of Prop.~\ref{prop:convex_structure}.\eqref{it:c_convex_structure}, also the hypothesis of Prop.~\ref{prop:general_equivalence} is satisfied by the sets $\OO$ and $\cOO$. We then obtain the following characterization of detection equivalence for two IWs.

\begin{proposition}\label{teo:equivalence}
If $\wit_1$, $\wit_2$ are two IWs on the set $\OO$ and $\wit_i(\Ao,\Bo) = \delta_i - \pair{(\Ao,\Bo)}{(F_i,G_i)}$ for all $(\Ao,\Bo)\in\OO$ and $i=1,2$, then $\wit_1\approx\wit_2$ if and only if
\equano{
F_2(x) = \alpha F_1(x) + A\,, \qquad\qquad G_2(y) = \alpha G_1(y) + B\,, \qquad\qquad \delta_2 = \alpha\delta_1 + \tr{A+B}
}
for some $\alpha > 0$ and $A,B\in\lhs$.
\end{proposition}
\begin{proof}
By Props.~\ref{prop:general_equivalence} and \ref{prop:convex_structure}.\eqref{it:c_convex_structure}, the equivalence $\wit_1\approx\wit_2$ holds if and only if, for some $\alpha > 0$, both the following conditions are satisfied:
\begin{enumerate}[-]
\item $(F_1,G_1) - \alpha (F_2,G_2) \in \V(\cOO)^\perp$ \ $\Leftrightarrow$ \ by \eqref{eq:annihilator_obs}, there exist $A,B\in\lhs$ such that $F_1(x) - \alpha F_2(x) =  A$ for all $x\in X$ and $G_1(y) - \alpha G_2(y) =  B$ for all $y\in Y$;
\item $\delta_1 - \alpha\delta_2 = \pair{(\Uo_X,\Uo_Y)}{(F_1,G_1) - \alpha (F_2,G_2)}$  \ $\Leftrightarrow$ \ $\delta_1 - \alpha\delta_2 = \tr{A} + \tr{B}$ with $A$ and $B$ given by the previous item.
\end{enumerate}
This concludes the proof.
\end{proof}

We have already seen in tha main paper that, up to detection equivalence, any tight IW can be associated with a state discrimination problem with postmeasurement information. This problem consists in discriminating the classical labels of some {\em partitioned state ensemble}, that is, a couple $\enp = (\en,\{X,Y\})$, in which
\begin{enumerate}[(i)]
\item $X$ and $Y$ are disjoint finite sets, and
\item $\en$ is a {\em state ensemble} with label set $X\cup Y$, i.e., an element of $\fun{X\cup Y}$ such that $\en(z)\geq 0$ for all $z\in X\cup Y$ and $\sum_{z\in X\cup Y} \tr{\en(z)} = 1$.
\end{enumerate}
Notice that the pairing \eqref{eq:iden_V*} for the restrictions $F_1 = \left.\en\right|_X$ and $G_1 = \left.\en\right|_Y$ rewrites
\equasi{eq:pairing_Ppg}{
\pair{(\Ao,\Bo)}{(\left.\en\right|_X,\left.\en\right|_Y)} = \Prg(\enp;\Ao,\Bo)  \qquad \forall (\Ao,\Bo)\in\OO \,,
}
where $\Prg(\enp;\Ao,\Bo)$ is the guessing probability with premeasurement information defined in~(4) of the main text. It follows that
\equasi{eq:link_P_delta}{
\delta_{\OO}(\left.\en\right|_X,\left.\en\right|_Y) = \Prg(\enp)\,, \qquad\qquad
\delta_{\cOO}(\left.\en\right|_X,\left.\en\right|_Y) = \Ppg(\enp) \,,
}
where $\Prg(\enp)$ and $\Ppg(\enp)$ are the optimal guessing probabilities with premeasurement and postmeasurement information given by~(5) and~(6) of the paper. In particular, whenever the strict inequality $\Ppg(\enp) < \Prg(\enp)$ holds, we can define the tight IW associated with the partitioned state ensemble $\enp$ as we did in~(7):
\equasi{eq:standard_IW_suppl}{
\wit_\enp(\Ao,\Bo) = \Ppg(\enp) - \Prg(\enp;\Ao,\Bo) \qquad \forall (\Ao,\Bo)\in\OO\,.
}

We remark that the evaluation of $\Prg(\enp)$ consists in solving two separate standard state discrimination problems: one for the subensemble $\en_X = p(X)^{-1}\left.\en\right|_X$ and another one for $\en_Y = p(Y)^{-1}\left.\en\right|_Y$. (Note: Although we use similar notations for subensembles and margin measurements, their correct interpretation is always clear from the context.) On the other hand, the optimization problem in the definition of $\Ppg(\enp)$ can be turned into a single standard state discrimination problem by means of \cite[Thm.~2]{SCaHeToPRA18}. Therefore, in order to evaluate both probabilities in \eqref{eq:link_P_delta}, one can resort to techniques from standard quantum state discrimination, as those described e.g.~in \cite[Sec.~IV\,B]{SCaHeToPRA18}, \cite{SHolevo73,SYuKeLa75,SQDET76,SBae13}.

If the partitioned state ensemble $\enp$ is not trivial, any compatible pair of measurements solving the optimization problem in the definition of $\Ppg(\enp)$ necessarily lies on the relative boundary of the set $\cOO$. More precisely, we have the following fact. It should be compared with the detailed characterization of the relative boundary and the extreme points of the compact convex set $\cOO$ provided in \cite{SGuCu18}.

\begin{proposition}\label{prop:boundary_Ppost}
Suppose the partitioned state ensemble $\enp = (\en,\{X,Y\})$ is such that either one of the restrictions $\left.\en\right|_X$ or $\left.\en\right|_Y$ is not a constant function. Then, for two compatible measurements $\Ao$ and $\Bo$, the equality $\Ppg(\enp) = \Prg(\enp;\Ao,\Bo)$ entails that $(\Ao,\Bo)\in\partial\cOO$. Moreover, there always exists some extreme point $(\Ao_0,\Bo_0)$ of $\cOO$ such that $\Ppg(\enp) = \Prg(\enp;\Ao_0,\Bo_0)$.
\end{proposition}
\begin{proof}
By combining Lemma \ref{lem:delta} and \eqref{eq:pairing_Ppg}, \eqref{eq:link_P_delta}, the second claim is always true, while the first one holds whenever $(\left.\en\right|_X \,,\,\left.\en\right|_Y)$ is not an element of $\V(\cOO)^\perp$. By \eqref{eq:annihilator_obs}, this is equivalent to either $\left.\en\right|_X$ or $\left.\en\right|_Y$ being not constant.
\end{proof}

Note that the last proposition does not require $\Ppg(\enp) < \Prg(\enp)$. Indeed, if two compatible measurements $\Ao$ and $\Bo$ attain both the equalities $\Ppg(\enp) = \Prg(\enp) = \Prg(\enp;\Ao,\Bo)$, then actually $(\Ao,\Bo)\in\partial\cOO\cap\partial\OO$ by an easy extension of the argument in the last proof.

Remarkably, as a consequence of Prop.~\ref{prop:boundary_Ppost}, whenever there is a unique pair of compatible measurements $\Ao$ and $\Bo$ attaining the equality $\Ppg(\enp) = \Prg(\enp;\Ao,\Bo)$, then $(\Ao,\Bo)$ is necessarily an extreme point of the compact convex set $\cOO$; some examples can be found in \cite[Secs.~V and VI]{SCaHeToPRA18}.

\section{Incompatibility witnesses with two mutually unbiased bases}\label{sec:MUB_suppl}

In this section, we suppose $\{\phii_h \mid h = 1,\ldots,d\}$ and $\{\psi_k\mid k = 1,\ldots,d\}$ are two fixed mutually unbiased bases (MUB) of the $d$-dimensional Hilbert space $\hh$. We will show how these bases can be used to construct a family of IWs for pairs of measurements with outcomes $X = \{(h,\phii)\mid h=1,\ldots,d\}$ and $Y = \{(k,\psi)\mid k=1,\ldots,d\}$. Moreover, as a byproduct of this construction, we will also characterize the amount of uniform noise that is needed in order to make the two given MUB compatible. Note that, although when one regards $X$ and $Y$ as separate sets, the extra symbols $\phii$ and $\psi$ are clearly redundant, nonetheless they are needed to define the disjoint union $Z=X\cup Y$, and then consider the partition $\{X,Y\}$ of $Z$.

\subsection{Construction of the IWs}\label{subsec:construction_MUB}

{For all $\bmu = (\mu_\phii,\mu_\psi)\in [0,1]\times [0,1]$, we can define the state ensemble $\en_\bmu$ with label set $Z$, given as
\equasi{eq:def_en_bmu}{
\en_\bmu (j,\ell) = \frac{1}{2d} \left[\mu_\ell\kb{\ell_j}{\ell_j} + (1-\mu_\ell) \frac{1}{d}\,\id\right] \qquad \forall j\in\{1,\ldots,d\},\,\ell\in\{\phii,\psi\}\,,
}
and the corresponding partitioned state ensemble $\enp_\bmu = (\en,\{X,Y\})$. For this state ensemble, all labels $z\in Z$ occur with the same probability $p(z) = \tr{\en_\bmu (z)} = 1/(2d)$; moreover, $p(X) = p(Y) = 1/2$ are the probabilities that a label occur in the set $X$ or $Y$, respectively. The subensemble $\en_{\bmu,X}(x) = (1/p(X))\en_\bmu(x) \ \ \forall x\in X$ is given by
\begin{equation}\label{eq:MUB_subensemble}
\en_{\bmu,X} (h,\phii) = \frac{1}{d} \left[\mu_\phii\kb{\phii_h}{\phii_h} + (1-\mu_\phii) \frac{1}{d}\,\id\right] \,,
\end{equation}
and it corresponds to a situation in which Alice randomly picks any of the vectors within the first basis with equal probability, and then she adds it uniform noise with intensity $1-\mu_\phii$. A similar interpretation holds for the subensemble $\en_{\bmu,Y}$, with a possibly different choice of the noise parameter $\mu_\psi$.}

{In order to obtain some more inequivalent IWs, we can let either $\mu_\phii$ or $\mu_\psi$ take even slightly negative values, in which case $\en_{\bmu,X}$ and $\en_{\bmu,Y}$ can no longer be interpreted as noisy ensembles. The most general choice still making \eqref{eq:def_en_bmu} define a state ensemble is indeed $\bmu = [1/(1-d),1]\times [1/(1-d),1]$. In the following, we always use the latter values, and we further assume ${\bmu} \neq (0,0)$ to avoid the trivial case.}

In order to use the state ensemble $\en_\bmu$ for constructing a tight IW as in \eqref{eq:standard_IW_suppl}, first of all we need to evaluate the pre- and postmeasurement guessing probabilities $\Prg(\enp_\bmu)$ and $\Ppg(\enp_\bmu)$. To this aim, we recall the following two useful results.

\begin{proposition}[Prop.~2 of \cite{SCaHeToPRA18}]\label{prop:easy}
Suppose $\en$ is a state ensemble with label set $X$. For all $x\in X$, denote by $\lambda(x)$ the largest eigenvalue of $\en(x)$, and by $\Pi(x)$ the orthogonal projection onto the $\lambda(x)$-eigenspace of $\en(x)$. Define
\begin{equation}\label{eq:def_lam_Xen}
\lambda_\en = \max_{x\in X} \lambda(x) \, , \qquad\qquad X_\en = \{ x \in X: \lambda(x) = \lambda_\en \} \, .
\end{equation}
Then, if there exists $\nu\in\Rb$ such that
\begin{equation}\label{eq:easy_condition_1}
\sum_{x\in X_\en} \Pi(x) = \nu\id \,,
\end{equation}
we have the following consequences:
\begin{enumerate}[(a)]
\item $\nu = \tfrac{1}{d} \sum_{x\in X_\en} \rank{\Pi(x)}$;\label{item:(mu)prop_easy}
\item $\Pg(\en) = d\lambda_\en$;\label{item:(lambda)prop_easy}
\item a (not necessarily unique) measurement $\Mo_0$ attaining the maximum guessing probability $\Pg(\en)$ is\label{item:(M0)prop_easy}
\begin{equation*}
\Mo_0(x) = \begin{cases}
\nu^{-1}\Pi(x) & \text{ if } x\in X_\en \\
0 & \text{ if } x\notin X_\en
\end{cases}\,.
\end{equation*}
\end{enumerate}
\end{proposition}

\begin{theorem}[Thm.~2 of \cite{SCaHeToPRA18}]\label{prop:equiv}
For any partitioned state ensemble $\enp = (\en,\{X,Y\})$, we have
\begin{equation*}
\Ppg(\enp) = (\abs{X}p(Y) + \abs{Y}p(X)) \Pg(\enf) \,,
\end{equation*}
where $\enf$ is the state ensemble having the Cartesian product $X\times Y$ as its label set, and given by
\begin{equation}\label{eq:assisting}
\enf(x,y) = \frac{\en(x) + \en(y)}{\abs{X}p(Y) + \abs{Y}p(X)} \qquad \forall (x,y)\in X\times Y \,.
\end{equation}
Moreover, for a measurement $\Mo : X\times Y\to\lhs$, we have the equivalence
$$
\Pg(\enf;\Mo) = \Pg(\enf) \quad \Leftrightarrow \quad \Prg(\enp;\Mo_X,\Mo_Y) = \Ppg(\enp) \,,
$$
where $\Mo_X$ and $\Mo_Y$ are the two margin measurements of $\Mo$ defined in \eqref{eq:defC}.
\end{theorem}

{We can immediately use Prop.~\ref{prop:easy} to evaluate the optimal guessing probability for the subensemble $\en_{\bmu,X}$. Indeed, the largest eigenvalue of \eqref{eq:MUB_subensemble} is}
\aligno{
\lam(h,\phii) & = \begin{cases}
\displaystyle\frac{1}{d^2}\left[1+(d-1)\mu_\phii\right] & \text{ if } \mu_\phii \geq 0 \\[0.3cm]
\displaystyle\frac{1}{d^2}\left(1-\mu_\phii\right) & \text{ if } \mu_\phii < 0 
\end{cases} \\[0.2cm]
& = \frac{1}{d^2}\left\{1+\frac{1}{2}\left[(d-2)\mu_\phii + d\abs{\mu_\phii}\right]\right\} \qquad \forall h=1,\ldots,d \,.
}
Thus, in \eqref{eq:def_lam_Xen}, the overall largest eigenvalue $\lam_{\en_{\bmu,X}}$ is given by the latter expression, and it is attained on the whole label set $X_{\en_{\bmu,X}} = X$. Eq.~\eqref{eq:easy_condition_1} is then easily verified, so that Prop.~\ref{prop:easy} can be applied. Hence, Prop.~\ref{prop:easy}.\eqref{item:(lambda)prop_easy} yields
$$
\Pg(\en_{\bmu,X}) = \frac{1}{d}\left\{1+\frac{1}{2}\left[(d-2)\mu_\phii + d\abs{\mu_\phii}\right]\right\} \,.
$$
A similar formula holds also for $\Pg(\en_{\bmu,Y})$. We finally obtain
\aligsi{
\Prg(\enp_\bmu) & = p(X)\Pg(\en_{\bmu,X}) + p(Y)\Pg(\en_{\bmu,Y}) \notag \\
& = \frac{1}{4}\left[2 + \abs{\mu_\phii} + \abs{\mu_\psi} + \left(1-\frac{2}{d}\right) (\mu_\phii + \mu_\psi - 2)\right] \,. \label{eq:Ppre_MUB}
}

Now we tackle the more difficult problem of evaluating $\Ppg(\enp_\bmu)$. In order to apply Thm.~\ref{prop:equiv}, we first write the auxiliary state ensemble \eqref{eq:assisting}, which is
\equasi{eq:assisting_MUB}{
\enf((h,\phii),(k,\psi)) = \frac{1}{2d^2} \left[\mu_\phii\kb{\phii_h}{\phii_h} + \mu_\psi\kb{\psi_k}{\psi_k} + \frac{1}{d} (2 - \mu_\phii - \mu_\psi) \,\id\right]
}
for the partitioned state ensemble $\enp_\bmu$. Then, we check if we can apply Prop.~\ref{prop:easy} in order to calculate $\Pg(\enf)$. To this aim, we need to find the spectral decomposition of \eqref{eq:assisting_MUB} for all $h,k$. The next lemma is useful to this purpose.

\begin{lemma}\label{lem:lemma}
Let $\varphi,\psi\in\hh$ be two unit vectors such that $\abs{\ip{\varphi}{\psi}}^2 = 1/d$. Denote $Q=\kb{\varphi}{\varphi}$ and $P=\kb{\psi}{\psi}$, and let
\begin{equation}\label{eq:defS}
S = qQ+pP \quad \text{with} \quad q,p\in\Rb \quad \text{and} \quad (q,p) \neq (0,0)\,.
\end{equation}
Then, the eigenvalues of the selfadjoint operator $S$ are
\begin{alignat}{2}
\lambda_+ & = \frac{1}{2}\left[(q+p)+\sqrt{q^2 + p^2 - 2\Delta qp}\right] & \qquad & \text{with multiplicity $\geq 1$}\,, \label{eq:lam+}\\
\lambda_- & = \frac{1}{2}\left[(q+p)-\sqrt{q^2 + p^2 - 2\Delta qp}\right] & \qquad & \text{with multiplicity $\geq 1$}\,, \label{eq:lam-}\\
\lambda_0 & = 0 & \qquad & \text{with multiplicity $\geq d-2$} \,, \label{eq:lam0}
\end{alignat}
where
\equano{
\Delta = 1-\frac{2}{d} \geq 0 \,.
}
They satisfy the following inequalities:
\begin{enumerate}[-]
\item if $q\geq 0$ and $p\geq 0$, then $\lambda_+ > \lambda_- \geq 0$, with equality if and only if $qp = 0$;
\item if either $q>0$ and $p<0$, or $q<0$ and $p>0$, then $\lambda_+ > 0 > \lambda_-$;
\item if $q\leq 0$ and $p\leq 0$, then $0 \geq \lambda_+ > \lambda_-$, with equality if and only if $qp = 0$.
\end{enumerate}
Moreover, the three selfadjoint operators
\begin{align}
\Pi_+ & = \frac{1}{\lam_+ - \lam_-} \left\{q Q + p P - \frac{d\lambda_-}{d-1} \left[Q+P-(QP+PQ)\right]\right\}\,, \label{eq:defPi+} \\
\Pi_- & = \frac{1}{\lambda_- - \lambda_+} \left\{q Q + p P  - \frac{d\lambda_+}{d-1} \left[Q+P-(QP+PQ)\right]\right\}\,, \label{eq:defPi-} \\
\Pi_0 & = \id - \frac{d}{d-1}\left[Q+P-(QP+PQ)\right] \label{eq:defPi0}
\end{align}
constitute an orthogonal resolution of the identity, with $S\Pi_k = \Pi_k S = \lam_k S$ for all $k\in\{+,-,0\}$, and $\rank{\Pi_+} = \rank{\Pi_-} = 1$, $\rank{\Pi_0} = d-2$.
\end{lemma}

\begin{proof}
After noticing that the vectors \(\{\varphi,\psi\}\) are linearly independent, we define the following two subspaces of $\hh$ with respective dimensions $2$ and $d-2$
$$
\hh_1 = \spanno{\varphi,\psi} \,,\qquad\qquad \hh_0 = \hh_1^\perp \, .
$$
Then, $S\hh_1\subseteq \hh_1$ and $S\hh_0 = \{0\}$. In particular, the eigenvalue $0$ has multiplicity greater than or equal to $d-2$ in $S$. The matrix form of the restriction $S_1 := \left. S\right|_{\hh_1}$ with respect to the (nonorthogonal) basis $\{\varphi,\psi\}$ of $\hh_1$ is
$$
S_1 = \left(\begin{array}{cc}
q & q \ip{\varphi}{\psi} \\
p \ip{\psi}{\varphi} & p 
\end{array}\right) \,.
$$
The roots of the characteristic polynomial of $S_1$ are the eigenvalues $\lambda_+$ and $\lambda_-$ in \eqref{eq:lam+}, \eqref{eq:lam-}. Since the quadratic form $q^2 + p^2 - 2\Delta qp$ is positive definite, the square root in \eqref{eq:lam+} and \eqref{eq:lam-} is nonzero, hence $\lambda_+ > \lambda_-$. The remaining inequalities involving $0$ are straightforward calculations. The multiplicities of $\lambda_+$ and $\lambda_-$ in $S$ can not be less than the respective multiplicities in $S_1$, which are $1$. This completes the proof of the statements about the eigenvalues.\\
Now, we claim that
\begin{equation}\label{eq:defPi1}\tag{$\sim$}
\Pi_1 = \frac{d}{d-1}[Q+P-(QP+PQ)]
\end{equation}
is the orthogonal projection onto $\hh_1$. Indeed,
\begin{enumerate}[-]
\item $\Pi_1^* = \Pi_1$ \quad (immediate);
\item $\left. \Pi_1 \right|_{\hh_0} = 0$ \quad (immediate);
\item $\Pi_1 Q = Q$ \quad (because $Q^2 = Q$ and $QPQ = (1/d)Q$) \quad $\Rightarrow$ \quad $\Pi_1\varphi = \varphi$;
\item $\Pi_1 P = P$ \quad (because $P^2 = P$ and $PQP = (1/d)P$) \quad $\Rightarrow$ \quad $\Pi_1\psi = \psi$.
\end{enumerate}
By applying the spectral theorem to $S_1$, there exist two rank-$1$ orthogonal projections $\Pi_+$ and $\Pi_-$ defined on $\hh$ and satisfying the relations
$$
\Pi_+ + \Pi_- = \Pi_1 \,, \qquad \qquad S = \lambda_+\Pi_+ + \lambda_-\Pi_- \,.
$$
Inserting the expressions \eqref{eq:defS} for $S$ and \eqref{eq:defPi1} for $\Pi_1$ into these relations, and solving with respect to $\Pi_+$, $\Pi_-$, we get \eqref{eq:defPi+} and \eqref{eq:defPi-}. Since the operator \eqref{eq:defPi0} is the orthogonal projection $\Pi_0 = \id-\Pi_1$ onto $\hh_0$, the proof of the last claim of the lemma is concluded.
\end{proof}

The expression in the square root of \eqref{eq:lam+}, \eqref{eq:lam-} is the quadratic form associated with the $2\times 2$ matrix
$$
G = \left(\begin{array}{cc}
1 & -\Delta \\
-\Delta & 1 
\end{array}\right) \,.
$$
Since $G$ is positive definite, this quadratic form actually is the squared norm
$$
\no{(q,p)}_G^2 = (q,p)G(q,p)^t = q^2 + p^2 - 2\Delta qp
$$
defined by the Euclidean scalar product $\pair{(q_2,p_2)}{(q_1,p_1)}_G = (q_1,p_1)G(q_2,p_2)^t$ in $\Rb^2$. Remarkably, in the special case $d=2$, we have $\Delta = 0$, hence $\pair{\cdot}{\cdot}_G$ is the canonical Euclidean scalar product of $\Rb^2$. Also note that in this case, independently of the values of $q$ and $p$, the largest eigenvalue of $S$ is $\lam_+$ with multiplicity $1$, its associated eigenprojection is $\Pi_+$, and the orthogonal projection $\Pi_0$ is zero.

Assuming $d\geq 3$, only for $q>0$ or $p>0$ the largest eigenvalue of $S$ is still given by $\lam_+$; otherwise, for $q\leq 0$ and $p\leq 0$, the largest eigenvalue of $S$ is $0$. In the former case, the eigenvalue $\lam_+$ has still multiplicity $1$ and associated eigenprojection $\Pi_+$; in the latter case, the eigenvalue $0$ has either multiplicity $d-2$ and associated eigenprojection $\Pi_0$ (subcase $qp\neq 0$), or multiplicity $d-1$ and associated eigenprojection $\Pi_0+\Pi_+$ (subcase $qp = 0$).

As a consequence of the last two paragraphs, the spectral decomposition of the operators \eqref{eq:assisting_MUB} is different according to the cases:
\begin{enumerate}[(C1)]
\item $d=2$ or $\max\{\mu_\phii,\mu_\psi\}>0$;\label{it:case1}
\item $d\geq 3$ and $\max\{\mu_\phii,\mu_\psi\}<0$;\label{it:case2}
\item $d\geq 3$ and $\max\{\mu_\phii,\mu_\psi\}=0$.\label{it:case3}
\end{enumerate}
More precisely, all the operators $\{\enf(x,y) \mid (x,y)\in X\times Y\}$ always have the same largest eigenvalue, which is
$$
\lam_\enf = \begin{cases}
\displaystyle\frac{1}{4d^2}\left[\left(1-\frac{2}{d}\right)(\mu_\phii+\mu_\psi) + \frac{4}{d} + \no{\bmu}_G\right] & \text{ in case (C\ref{it:case1})} \\[0.3cm]
\displaystyle\frac{1}{2d^3}\left[2 - (\mu_\phii+\mu_\psi)\right] & \text{ in cases (C\ref{it:case2})-(C\ref{it:case3})}
\end{cases}
$$
by Lemma \ref{lem:lemma}; in particular, the state ensemble $\enf$ attains its largest eigenvalue on the whole label set $X\times Y$, that is, $(X\times Y)_\enf = X\times Y$ with the notation of Prop.~\ref{prop:easy}. Moreover, combining \eqref{eq:defPi+} and \eqref{eq:defPi0} according to the case at hand, the orthogonal projection onto the $\lam_\enf$-eigenspace of \eqref{eq:assisting_MUB} is
\begin{subequations}\label{eq:defPi}
\equasi{eq:defPi_1}{
\Pi((h,\phii),(k,\psi)) = a \id + b (\mu_\phii \Qo(h) + \mu_\psi \Po(k)) - \frac{dc}{d-1} \left[\Qo(h)+\Po(k)-(\Qo(h)\Po(k)+\Po(k)\Qo(h))\right] \,,
}
where
\begin{equation}
\Qo(h) = \kb{\phii_h}{\phii_h}\,, \qquad\qquad \Po(k) = \kb{\psi_k}{\psi_k}
\end{equation}
and
\equasi{eq:def_abc}{
\begin{aligned}
a & = 0\,,  \qquad & b & = \frac{1}{\no{\bmu}_G}\,, \qquad & c & = \frac{1}{2}\left(\frac{\mu_\phii + \mu_\psi}{\no{\bmu}_G} - 1\right) \qquad && \text{in case (C\ref{it:case1})} \\
a & = 1\,,  \qquad &  b & = 0\,, \qquad & c & = 1 \qquad && \text{in case (C\ref{it:case2})} \\
a & = 1\,,  \qquad & b & = \frac{1}{\no{\bmu}_G}\,, \qquad & c & = 0 \qquad && \text{in case (C\ref{it:case3})}\,.
\end{aligned}
}
\end{subequations}
Since $\sum_{h=1}^d\Qo(h) = \sum_{k=1}^d\Po(k) = \id$ and the coefficients $a,b,c$ do not depend on $h,k$, we have
$$
\sum_{(x,y)\in (X\times Y)_\enf} \Pi(x,y) = \sum_{h,k=1}^d \Pi((h,\phii),(k,\psi)) = \nu \id
$$
where
\equasi{eq:nu}{
\nu = d^2 a + db(\mu_\phii + \mu_\psi) - 2dc =
\begin{cases}
d & \text{ in case (C\ref{it:case1})}\\
d(d-2) & \text{ in case (C\ref{it:case2})}\\
d(d-1) & \text{ in case (C\ref{it:case3})}
\end{cases} \,.
}
In all cases, Eq.~\eqref{eq:easy_condition_1} holds, hence we can apply Prop.~\ref{prop:easy} to determine $\Pg(\enf) = d\lam_\enf$ and the optimal measurement $\Mo_0(x,y) = \nu^{-1}\Pi(x,y) \ \ \forall (x,y)\in X\times Y$. Note that all projections $\Pi(x,y)$ have the same rank, hence $\rank{\Mo_0(x,y)} = \nu/d$ for all $x,y$ by Prop.~\ref{prop:easy}.\eqref{item:(mu)prop_easy}.

Now that we have found $\Ppg(\enf)$ and characterized a measurement attaining it, by making use of Thm.~\ref{prop:equiv} we can translate these results into the state discrimination problem with postmeasurement information for the partitioned state ensemble $\enp_\bmu$. The maximal guessing probability in the postmeasurement scenario is $\Ppg(\enp_\bmu) = d\Pg(\enf) = d^2 \lam_\enf$; it is attained on the compatible pair given by the margins of the measurement $\Mo_0 = \nu^{-1}\Pi$, which are easily calculated from \eqref{eq:defPi} and the normalization of $\Qo$ and $\Po$.
We collect the conclusions in the next proposition.

\begin{proposition}\label{prop:summary_Ppost_MUB}
For all $\bmu = (\mu_\phii,\mu_\psi) \in [1/(1-d),1]\times [1/(1-d),1]$ with $\bmu\neq (0,0)$, we have
\equasi{eq:Ppost_MUB}{
\Ppg(\enp_\bmu) = \begin{cases}
\displaystyle \frac{1}{4}\left[\left(1-\frac{2}{d}\right)(\mu_\phii+\mu_\psi) + \frac{4}{d} + \no{\bmu}_G\right] & \text{ in case (C\ref{it:case1})} \\[0.3cm]
\displaystyle \frac{1}{2d}\left[2 - (\mu_\phii+\mu_\psi)\right] & \text{ in cases (C\ref{it:case2})-(C\ref{it:case3})}
\end{cases} \,.
}
Moreover, $\Ppg(\enp_\bmu) = \Prg(\enp_\bmu;\Mo_{0,X},\Mo_{0,Y})$, where the measurement $\Mo_0 = \nu^{-1}\Pi$ is expressed in terms of the constant $\nu$ defined in \eqref{eq:nu} and the rank-$(\nu/d)$ orthogonal projections $\Pi(x,y)$ given by \eqref{eq:defPi} (which both depend on $\bmu$). Explicitly,
\equasi{eq:optimargins}{
\Mo_{0,X}(h,\phii) = \gamma_\phii \Qo(h) + (1-\gamma_\phii)\frac{1}{d}\,\id \,, \qquad\qquad \Mo_{0,Y}(k,\psi) = \gamma_\psi \Po(k) + (1-\gamma_\psi)\frac{1}{d}\,\id \,,
}
where
\begin{equation}\label{eq:gamma}
\gamma_\ell = 
\frac{d}{\nu}\left[b\mu_\ell - \frac{(d-2)c}{d-1}\right] =
\begin{cases}
\displaystyle \frac{d\mu_\ell - (d-2)(\mu_{\overline{\ell}} - \no{\bmu}_G)}{2(d-1)\no{\bmu}_G} & \text{ in case (C\ref{it:case1})} \\[0.3cm]
\displaystyle \frac{1}{1-d} (1-\delta_{\mu_\ell,0}) & \text{ in cases (C\ref{it:case2})-(C\ref{it:case3})}
\end{cases} \,.
\end{equation}
In the last equation, $(\ell,\overline{\ell})$ denotes either $(\phii,\psi)$ or $(\psi,\phii)$; moreover, $\delta_{\mu_\phii,0}$ and $\delta_{\mu_\psi,0}$ are the usual Kronecker deltas.
\end{proposition}

Remarkably, in case (C\ref{it:case2}) there exists an optimal pair of compatible measurements which is unaffected by the value of $\bmu$; it is the pair $(\Ao,\Bo)$ given by $\Ao(h,\phii) = [1/(1-d)]\Qo(h) + [1/(d-1)] \id$ and $\Bo(k,\psi) = [1/(1-d)]\Po(h) + [1/(d-1)] \id$ for all $h,k$. Note that $(\Ao,\Bo)$ is optimal also in case (C\ref{it:case3}), where the margins of $\Mo_0$ however yield a different solution.

We can now determine the values of $\bmu$ for which $\Ppg(\enp_\bmu) < \Prg(\enp_\bmu)$, and for these values explicitly evaluate the tight IW \eqref{eq:standard_IW_suppl} associated with $\enp_\bmu$. This yields the first main result of the present section.

\begin{theorem}[Thm.~3 of the main paper]\label{thm:MUBwitness_suppl}
Let $\bmu = (\mu_\phii,\mu_\psi) \in [1/(1-d),1]\times [1/(1-d),1]$ with $\bmu\neq (0,0)$. Then, we have the strict inequality $\Ppg(\enp_\bmu) < \Prg(\enp_\bmu)$ if and only if $\mu_\phii\mu_\psi\neq 0$ and either $d=2$ or $\max\{\mu_\phii,\mu_\psi\} > 0$. In this case, the tight IW associated with the partitioned state ensemble $\enp_\bmu$ by means of \eqref{eq:standard_IW_suppl} is
\begin{equation}\label{eq:MUBwitness_suppl}
\begin{aligned}
\wit_{\enp_\bmu} (\Ao,\Bo) = &  \frac{1}{4}\left(\mu_\phii + \mu_\psi + \no{\bmu}_G\right) - \frac{1}{2d} \left\{ \mu_\phii \sum_{h=1}^d \tr{\Ao(h,\phii)\Qo(h)} + \mu_\psi \sum_{k=1}^d \tr{\Bo(k,\psi)\Po(k)} \right\} \,.
\end{aligned}
\end{equation}
Finally, the ensembles $\enp_\bmu$ and $\enp_\bnu$ yield detection equivalent IWs if and only if $\bnu = \alpha\bmu$ for some $\alpha > 0$.
\end{theorem}

\begin{proof}
By \eqref{eq:Ppre_MUB} and \eqref{eq:Ppost_MUB},
$$
\Prg(\enp_\bmu) - \Ppg(\enp_\bmu) =
\begin{cases}
\displaystyle \frac{1}{4}\left(\abs{\mu_\phii} + \abs{\mu_\psi} - \no{\bmu}_G\right) & \text{ in case (C\ref{it:case1})} \\[0.3cm]
\displaystyle 0 & \text{ in cases (C\ref{it:case2})-(C\ref{it:case3})}
\end{cases}\,.
$$
The triangular inequality for the norm $\no{\cdot}_G$ implies that the expression for case (C\ref{it:case1}) strictly positive unless $\mu_\phii = 0$ or $\mu_\psi = 0$. This proves the first claim.\\
Eq.~\eqref{eq:MUBwitness_suppl} then follows by combining \eqref{eq:Ppost_MUB} with
\aligno{
\Prg(\enp_\bmu;\Ao,\Bo) & = p(X)\Pg(\en_{\bmu,X};\Ao) + p(Y)\Pg(\en_{\bmu,Y};\Bo) \\
& = \frac{1}{2d}\left\{\mu_\phii\sum_h\tr{\Ao(h,\phii)\Qo(h)} + \mu_\psi\sum_k\tr{\Bo(k,\psi)\Po(k)} + 2 - (\mu_\phii + \mu_\psi)\right\} \,.
}
Finally, we prove the equivalence statement. If $\bnu = \alpha\bmu$ with $\alpha > 0$, then $\wit_{\enp_\bnu}=\alpha\wit_{\enp_\bmu}$ by \eqref{eq:MUBwitness_suppl}, hence the two witnesses are detection equivalent. Conversely, if $\wit_{\enp_\bmu}\approx\wit_{\enp_\bnu}$, then by Prop.~\ref{teo:equivalence} there exist $\alpha > 0$ and $A,B\in\lhs$ such that
$$
\en_\bnu (h,\phii) = \alpha \en_\bmu (h,\phii) + A \quad \forall h \quad \text{and} \quad \en_\bnu (k,\psi) = \alpha \en_\bmu (k,\psi) + B \quad \forall k \,.
$$
In particular, $\mu_\ell = 0 \ \Leftrightarrow \ \nu_\ell =0$, as any of the two equalities imply that both the maps $\en_\bmu (\cdot,\ell)$ and $\en_\bnu (\cdot,\ell)$ are constant.
Since
$$
\en_{\bnu} (j,\ell) = \frac{\nu_\ell}{\mu_\ell} \, \en_{\bmu} (j,\ell) + \frac{1}{2d^2}\left(1 - \frac{\nu_\ell}{\mu_\ell}\right)\,\id \qquad\forall j,\ell \qquad \left(\text{with } \frac{0}{0} := 0\right)\,,
$$
it must be $\alpha = \nu_\phii/\mu_\phii = \nu_\psi/\mu_\psi$ and $A = B = [(1-\alpha)/(2d^2)]\,\id$. In particular, $\bnu = \alpha\bmu$.
\end{proof}

\subsection{Incompatibility of noisy MUB}\label{subsec:incoMUB}

As a byproduct, the evaluation of $\Ppg(\enp_\bmu)$ and the solution to the related optimization problem provided in the last subsection allow to characterize the compatibility of the following pair of  measurements $(\Ao_\bgam,\Bo_\bgam)\in\OO$
\equasi{eq:def_AB_suppl}{
\Ao_\bgam(h,\phii) = \gamma_\phii \Qo(h) + (1-\gamma_\phii) \frac{1}{d}\,\id\,, \qquad\qquad \Bo_\bgam(k,\psi) = \gamma_\psi \Po(k) + (1-\gamma_\psi) \frac{1}{d}\,\id \,,
}
where $\bgam = (\gamma_\phii,\gamma_\psi)\in [1/(1-d),1]\times [1/(1-d),1]$. This set constitues all the values of the parameters $\bgam$ such that the operator valued maps $\Ao_\bgam,\Bo_\bgam$ constitute two POVMs. Note that, for $\gamma_\phii,\gamma_\psi \geq 0$, the measurements \eqref{eq:def_AB_suppl} can be interpreted as uniformly noisy versions of the sharp measurements $\Qo$ and $\Po$ \cite{SCaHeTo12}.

Since $\Prg(\enp_{\alpha\bmu};\Ao,\Bo) = \alpha\Prg(\enp_\bmu;\Ao,\Bo) + (1-\alpha)/d$ for $\alpha >0$, the optimization problem is the same for the partitioned state ensembles $\enp_\bmu$ and $\enp_{\alpha\bmu}$; hence, there is no restriction in parametrizing the family of state ensembles \eqref{eq:def_en_bmu} by means of the direction of the vector $\bmu$. We then choose the parametrization $\bmu = \bmu(\theta)$, given by
\equano{
\mu_\phii(\theta) = \alpha\left(\sqrt{d} \cos\theta + \sqrt{\frac{d}{d-1}}\sin\theta\right)\,, \qquad\qquad
\mu_\psi(\theta) = \alpha\left(\sqrt{d} \cos\theta - \sqrt{\frac{d}{d-1}}\sin\theta\right)
}
in terms of the single real parameter $\theta\in\left[-\pi\,,\,\pi\right]$. Here, $\alpha > 0$ is fixed in such a way that $\abs{\mu_\phii(\theta)} \leq 1/(d-1)$ and $\abs{\mu_\psi(\theta)} \leq 1/(d-1)$ for all $\theta$'s; for example, $\alpha = 1/(d\sqrt{d-1})$. The curve $\{\bmu(\theta)\mid \theta\in\left[-\pi\,,\,\pi\right]\}$ is an ellipse centered at $(0,0)$, and thus it spans all directions in the $(\gamma_\phii, \gamma_\psi)$-plane. With this parametrization, we always have $\no{\bmu(\theta)}_G = 2\alpha$.

In order to apply the classification into cases (C\ref{it:case1}), (C\ref{it:case2}) and (C\ref{it:case3}) of the the previous section, note that the inequality $\mu_\psi(\theta) < 0$ holds if and only if $\theta\in \left[-\pi\,,\,\pi\right]\setminus\left[-\theta_0 \,,\,\pi-\theta_0\right]$, where
\equasi{eq:def_theta0}{
\theta_0 = \pi-\arctan\sqrt{d-1}\in\left(\pi/2 \,,\, 3\pi/4\right] \,;
}
in addition, the set $\left\{-\theta_0 \,,\, \pi-\theta_0\right\}$ constitutes all the solutions to the equation $\mu_\psi(\theta) = 0$. Combining these facts with the symmetry $\mu_\phii(\theta) = \mu_\psi(-\theta)$, we have $\max\{\mu_\phii(\theta)\,,\,\mu_\psi(\theta)\} \geq 0$ if and only if $\theta\in\left[-\theta_0\,,\,\theta_0\right]$. Moreover, $\{ - \theta_0\,,\,\theta_0-\pi\,,\,\pi-\theta_0\,,\,\theta_0\}$ are all the values of $\theta$ for which either $\mu_\phii(\theta) = 0$ or $\mu_\psi(\theta) = 0$. Among the latter values, we have $\max\{\mu_\phii(\theta),\mu_\psi(\theta)\} = 0$ if and only if $\theta\in\{-\theta_0\,,\,\theta_0\}$. Thus, the cases described in the previous section occur as follows:
\begin{itemize}
\item[(C\ref{it:case1})] $\Leftrightarrow$ $d=2$ or $\theta\in (-\theta_0\,,\,\theta_0)$;
\item[(C\ref{it:case2})] $\Leftrightarrow$ $d\geq 3$ and $\theta\in [-\pi\,,\,\pi] \setminus [-\theta_0\,,\,\theta_0]$;
\item[(C\ref{it:case3})] $\Leftrightarrow$ $d\geq 3$ and $\theta\in \{-\theta_0\,,\,\theta_0\}$.
\end{itemize}

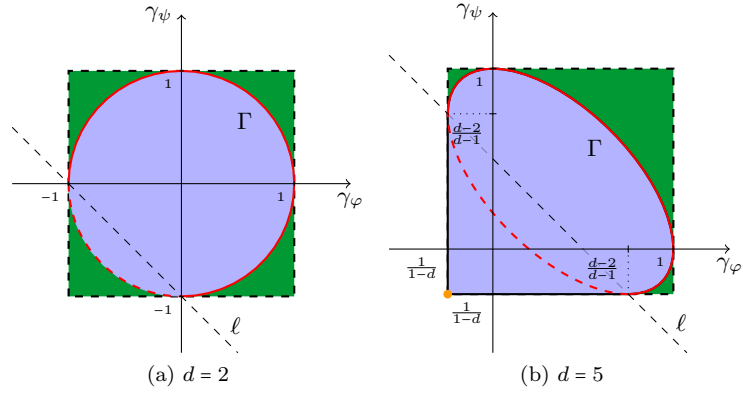
\begin{figure}[h!]\def\xy{1.5}
\centering
\subfloat[$d=2$]{
\begin{tikzpicture}[xscale=1.50,yscale=1.50,
declare function={
rsin(\x) = sin(deg(\x));
rcos(\x) = cos(deg(\x));
}]

\draw [thick, dashed, fill = darkgreen]
(-1,-1) -- (-1,1) -- plot (1,1) -- (1,-1) -- (-1,-1) ;

\draw [color=blue!30!white, domain=-1:1, samples=40, fill = blue!30!white]
plot ({rcos(\x*pi)}, {rsin(\x*pi)});

\draw [thick, red, domain=-0.5:1.0, samples=40]
plot ({rcos(\x*pi)}, {rsin(\x*pi)});

\draw [thick, dashed, red, domain=-1.0:-0.5, samples=40]
plot ({rcos(\x*pi)}, {rsin(\x*pi)});

\draw ({rcos(0.25*pi)}, {rsin(0.25*pi)}) node[anchor=north east]{$\Gamma$};

\draw[->] ({-\xy},0)--(\xy,0)node[anchor=north]{\small $\gamma_\phii$};
\draw[->] (0,{-\xy})--(0,\xy)node[anchor=east]{\small $\gamma_\psi$};

\draw [dashed] ({-\xy},{\xy-1}) -- ({\xy-1},{-\xy})
node[pos=0.99,above=0.1cm]{$\ell$};

\draw (-1,0)node[anchor=north east]{\tiny $-1$};
\draw[thin,-] (-1,-0.02)--(-1,0.02);

\draw (1,0)node[anchor=north east]{\tiny $1$};
\draw[thin,-] (1,-0.02)--(1,0.02);

\draw (0,-1)node[anchor=north east]{\tiny $-1$};
\draw[thin,-] (-0.02,-1)--(0.02,-1);

\draw (0,1)node[anchor=north east]{\tiny $1$};
\draw[thin,-] (-0.02,1)--(0.02,1);

\end{tikzpicture}}
\subfloat[$d=5$]{
\def\d{5}
\begin{tikzpicture}[xscale=2.4,yscale=2.4,
declare function={
ratan(\x) = rad(atan(\x));
rsin(\x) = sin(deg(\x));
rcos(\x) = cos(deg(\x));
}]

\def\th0{(pi-ratan(sqrt(\d-1)))}
\def\xymen{(-(0.46/1.2)*\xy)}
\def\xypiu{((1.05/1.2)*\xy)}

\draw [thick, dashed, fill = darkgreen]
({1/(1-\d)},{1/(1-\d)}) -- ({1/(1-\d)},1) -- (1,1) -- (1,{1/(1-\d)}) -- ({1/(1-\d)},{1/(1-\d)}) ;

\draw [thick, domain=-1:1, samples=40, fill=blue!30!white]
plot({1/(1-\d)},{1/(1-\d) + 1/2 * (\x+1)})
-- plot ({(\d*rcos(\x*\th0 + \th0) + 2 - \d)/(2*(1 - \d))}, {(\d*rcos(\x*\th0 - \th0) + 2 - \d)/(2*(1 - \d))})
-- plot({1/(1-\d) + 1/2 * (-\x+1)},{1/(1-\d)}) ;

\draw [thick,red, domain=-1:1, samples=40]
plot ({(\d*rcos(\x*\th0 + \th0) + 2 - \d)/(2*(1 - \d))}, {(\d*rcos(\x*\th0 - \th0) + 2 - \d)/(2*(1 - \d))}) ;

\draw [thick, dashed, red, domain=1:{(2*pi)/\th0-1}, samples=40]
plot ({(\d*rcos(\x*\th0 + \th0) + 2 - \d)/(2*(1 - \d))}, {(\d*rcos(\x*\th0 - \th0) + 2 - \d)/(2*(1 - \d))}) ;

\draw[->] ({\xymen},0)--({\xypiu},0)node[anchor=north]{\small $\gamma_\phii$};
\draw[->] (0,{\xymen})--(0,{\xypiu})node[anchor=east]{\small $\gamma_\psi$};

\draw [dashed] ({\xymen},{-\xymen+(\d-3)/(\d-1)}) -- ({-1/(\d-1)},{(\d-2)/(\d-1)}) ;
\draw [dashed,opacity=0.3] ({-1/(\d-1)},{(\d-2)/(\d-1)}) -- (0,{(\d-3)/(\d-1)});
\draw [dashed] (0,{(\d-3)/(\d-1)}) -- ({(\d-3)/(\d-1)},0) ;
\draw [dashed,opacity=0.3] ({(\d-3)/(\d-1)},0) -- ({(\d-2)/(\d-1)},{-1/(\d-1)});
\draw [dashed] ({(\d-2)/(\d-1)},{-1/(\d-1)}) -- ({-\xymen+(\d-3)/(\d-1)},{\xymen})
node[pos=0.90,above=0.1cm]{$\ell$} ;

\draw ({1/(1-\d)},0)node[anchor=north east]{\tiny $\frac{1}{1-d}$};
\draw[thin,-] ({1/(1-\d)},-0.02)--({1/(1-\d)},0.02);

\draw ({(\d-2)/(\d-1)},0)node[anchor=north east]{\tiny $\frac{d-2}{d-1}$};
\draw[thin,-] ({(\d-2)/(\d-1)},-0.02)--({(\d-2)/(\d-1)},0.02);
\draw[dotted,-] ({(\d-2)/(\d-1)},0)--({(\d-2)/(\d-1)},{1/(1-\d)});

\draw (1,0)node[anchor=north east]{\tiny $1$};
\draw[thin,-] (1,-0.02)--(1,0.02);

\draw (0,{1/(1-\d)})node[anchor=north east]{\tiny $\frac{1}{1-d}$};
\draw[thin,-] (-0.02,{1/(1-\d)})--(0.02,{1/(1-\d)});

\draw (0,{(\d-2)/(\d-1)})node[anchor=north east]{\tiny $\frac{d-2}{d-1}$};
\draw[thin,-] (-0.02,{(\d-2)/(\d-1)})--(0.02,{(\d-2)/(\d-1)});
\draw[dotted,-] (0,{(\d-2)/(\d-1)})--({1/(1-\d)},{(\d-2)/(\d-1)});

\draw (0,1)node[anchor=north east]{\tiny $1$};
\draw[thin,-] (-0.02,1)--(0.02,1);

\draw[color=darkyellow,fill=darkyellow] ({1/(1-\d)},{1/(1-\d)}) circle[radius=0.6pt];

\draw ({(\d*rcos(\th0) + 2 - \d)/(2*(1 - \d))}, {(\d*rcos(-\th0) + 2 - \d)/(2*(1 - \d))}) node[anchor=north east]{$\Gamma$};

\end{tikzpicture}}
\caption{The set of $\bgam = (\gamma_\phii,\gamma_\psi)$ for which \eqref{eq:def_AB_suppl} defines two measurements (green square), and the one for which these measurements are compatible (blue region) for different values of the dimension $d$. The red line is the ellipse \eqref{eq:ellipse}, whose solid part is the curve \eqref{eq:boundary_suppl} for $\theta\in\left[-\theta_0\,,\,\theta_0\right]$. The dashed line $\ell$ is described by the equation $\gamma_\phii + \gamma_\psi = (d-3)/(d-1)$. Note the symmetry around the origin in dimension $d=2$, and the compatible pair corresponding to $\bgam_0 = (1/(1-d)\,,\,1/(1-d))$ when $d\geq 3$ (orange dot).\label{fig:MUBregion_suppl}}
\end{figure}

If $\theta\in\left(-\theta_0\,,\,\theta_0\right)$, inserting $\mu_\phii(\theta)$ and $\mu_\psi(\theta)$ into \eqref{eq:gamma}, we obtain that the two  measurements \eqref{eq:def_AB_suppl} coincide with the margin measurements \eqref{eq:optimargins} for $\bgam = \bgam(\theta) = (\gamma_\phii(\theta),\gamma_\psi(\theta))$, where
\equasi{eq:boundary_suppl}{
\gamma_\phii(\theta) = \frac{d-2-d\cos(\theta+\theta_0)}{2(d-1)}\,, \qquad\qquad \gamma_\psi(\theta) = \frac{d-2-d\cos(\theta-\theta_0)}{2(d-1)} \,;
}
moreover, the equality $\Ppg(\enp_{\bmu(\theta)}) = \Prg(\enp_{\bmu(\theta)};\Ao_{\bgam(\theta)},\Bo_{\bgam(\theta)})$ holds with this choice of $\bgam$. By Prop.~\ref{prop:boundary_Ppost}, this implies that $(\Ao_{\bgam(\theta)},\Bo_{\bgam(\theta)})\in\partial\cOO$ for all $\theta\in(-\theta_0\,,\,\theta_0)$. Since the set $\partial\cOO$ is closed, an easy continuity argument shows that the last inclusion is true also for $\theta\in\{-\theta_0\,,\,\theta_0\}$ (although $\Ao_{\bgam(\theta)}$ and $\Bo_{\bgam(\theta)}$ do no longer coincide with the margins \eqref{eq:optimargins} for these values of $\theta$). Note that for $\lam>1$ the two measurements $\Ao_{\lam\bgam(\theta)}$ and $\Bo_{\lam\bgam(\theta)}$ are necessarily incompatible, as otherwise we would get the contradiction $(\Ao_{\bgam(\theta)},\Bo_{\bgam(\theta)}) = (1-1/\lam)(\Uo_X,\Uo_Y) + (1/\lam)(\Ao_{\lam\bgam(\theta)},\Bo_{\lam\bgam(\theta)})\in\ri{\cOO}$ by Prop.~\ref{prop:convex_structure}.\eqref{it:c_convex_structure} and \cite[Thm.~6.1]{SRock}.

In the $(\gamma_\phii, \gamma_\psi)$-plane, the curve $\Gamma = \{\bgam(\theta)\mid\theta\in\left[-\theta_0\,,\,\theta_0\right]\}$ given by \eqref{eq:boundary_suppl} is the part of the ellipse
\equasi{eq:ellipse}{
d(\gamma_\phii^2 + \gamma_\psi^2) + 2(d-2) \gamma_\phii \gamma_\psi - 2(d-2) (\gamma_\phii + \gamma_\psi) = 4-d
}
lying above  the line $\gamma_\phii + \gamma_\psi = (d-3)/(d-1)$, as depicted in Fig.~\ref{fig:MUBregion_suppl}. As we have just seen, all pairs of measurements corresponding to points beyond $\Gamma$ are incompatible.

If $\theta\in\left[-\pi\,,\,\pi\right]\setminus\left[-\theta_0\,,\,\theta_0\right]$, two different situations occur according to the dimension $d$. If $d=2$, inserting $\mu_\phii(\theta)$ and $\mu_\psi(\theta)$ into \eqref{eq:gamma} still gives the same results as before also for the new values of $\theta$. The curve \eqref{eq:boundary_suppl} for $\theta\in\left[-\pi\,,\,\pi\right]$ is then the unit circle in the $(\gamma_\phii, \gamma_\psi)$-plane. Since for $\bgam = (0,0)$ the two observables $\Ao_\bgam = \Uo_X$ and $\Bo_\bgam = \Uo_Y$ are trivially compatible, by convexity we thus conclude that the measurements $\Ao_\bgam$ and $\Bo_\bgam$ are compatible if and only if $\gamma_\phii^2 + \gamma_\psi^2 \leq 1$. On the other hand, if $d\geq 3$, then inserting $\mu_\phii(\theta)$ and $\mu_\psi(\theta)$ into \eqref{eq:gamma} yields $\gamma_\phii = \gamma_\psi = 1/(1-d)$ irrespectively of the value of $\theta$. Setting $\bgam_0 = (1/(1-d)\,,\,1/(1-d))$, it follows that $(\Ao_{\bgam_0},\Bo_{\bgam_0})\in\partial\cOO$, and, again by convexity, all $\bgam$'s lying between $\bgam_0$ and the curve $\Gamma$ correspond to compatible pairs of measurements.

The previous discussion is summarized in the second main result of this section.

\begin{theorem}[Thm.~4 of the main paper]\label{thm:2nd_thm}
\begin{enumerate}[(a)]
\item Suppose $d=2$. For $\bgam\in [-1,1]\times [-1,1]$, the two measurements $\Ao_\bgam$ and $\Bo_\bgam$ of \eqref{eq:def_AB_suppl} are compatible if and only if $\gamma_\phii^2 + \gamma_\psi^2 \leq 1$. Moreover, $(\Ao_\bgam,\Bo_\bgam)\in\partial\cOO$ if and only if $\gamma_\phii^2 + \gamma_\psi^2 = 1$.\label{it:2nd_thm_a}
\item Suppose $d\geq 3$. For $\bgam\in [1/(1-d)\,,\,1]\times [1/(1-d)\,,\,1]$, the two measurements $\Ao_\bgam$ and $\Bo_\bgam$ of \eqref{eq:def_AB_suppl} are compatible if and only if
\equano{
\gamma_\phii + \gamma_\psi \leq \frac{d-3}{d-1} \qquad
\text{or} \qquad d(\gamma_\phii^2 + \gamma_\psi^2) + 2(d-2) \gamma_\phii \gamma_\psi - 2(d-2) (\gamma_\phii + \gamma_\psi) \leq 4-d \,.
}
Moreover, $(\Ao_\bgam,\Bo_\bgam)\in\partial\cOO$ if and only if\label{it:2nd_thm_b}
\equano{
\begin{gathered}
\bgam = \left( \frac{t}{d-1} \,,\, \frac{1}{1-d} \right) \quad \text{or} \quad \bgam = \left( \frac{1}{1-d} \,,\, \frac{t}{d-1} \right) \quad \text{for some } t\in\left[-1 \,,\, d-2\right] \\
\text{or} \quad \bgam = \left( \frac{d-2-d\cos(\theta+\theta_0)}{2(d-1)} \,,\, \frac{d-2-d\cos(\theta-\theta_0)}{2(d-1)} \right) \quad \text{for some $\theta \in[-\theta_0,\theta_0]$} \,,
\end{gathered}
}
where $\theta_0$ is given in \eqref{eq:def_theta0}.
\end{enumerate}
\end{theorem}

Statement \eqref{it:2nd_thm_a} of the previous theorem is well known \cite{SBusch86}. On the other hand, statement \eqref{it:2nd_thm_b} was proved in the particular case in which the two measurements $\Ao_\bgam$ and $\Bo_\bgam$ are conjugate by the Fourier transform of the cyclic group $\Zb_d$, and restricting only to $\bgam\in [0,1]\times [0,1]$ \cite{SCaHeTo12}. Compared with the group theoretical approach of \cite{SCaHeTo12}, the present derivation of Thm.~\ref{thm:2nd_thm} has the advantage of not requiring any symmetry condition on the two MUB $\{\phii_1,\ldots,\phii_d\}$ and $\{\psi_1,\ldots,\psi_d\}$; thus, for dimensions $d\geq 4$ it actually applies to many inequivalent pairs of MUB, and not only to the Fourier conjugate pairs considered in \cite{SCaHeTo12} (see \cite{STaZy06} for a list of inequivalent pairs in dimensions $4\leq d \leq 16$ and the proof of the equivalence of all pairs in dimensions $d=2,3$). In the general (not symmetric) case, a proof of statement \eqref{it:2nd_thm_b} for $\bgam$ constrained on the diagonal $\gamma_\phii = \gamma_\psi$ is contained in \cite{SUoLuMoHe16,SDeSkFrBr18}.

We finally remark that the essential differences that lead to separate results for the cases $d=2$ and $d\geq 3$ are: (1) the additional symmetry $(\Ao_\bgam,\Bo_\bgam)\in\cOO \ \Leftrightarrow \ (\Ao_{-\bgam},\Bo_{-\bgam})\in\cOO$, which is specific of the $d=2$ case; (2) the fact that, for $\bgam_0 = (1/(1-d)\,,\,1/(1-d))$, the two measurements $(\Ao_{\bgam_0},\Bo_{\bgam_0})$ are compatible if and only if $d\geq 3$.




\end{document}